\RequirePackage{silence}
\documentclass[a4paper,onecolumn, superscriptaddress,accepted=2022-09-12]{quantumarticle}
\pdfoutput=1
\usepackage[T1]{fontenc} 
\usepackage[english]{babel} 
\usepackage[utf8]{inputenc}
\usepackage[numbers,sort&compress]{natbib}
\usepackage{etoolbox}

\usepackage[table,usenames,dvipsnames]{xcolor}
\usepackage[colorlinks=true,
			hypertexnames=false
			]{hyperref}
\PassOptionsToPackage{linktocpage}{hyperref} 
\usepackage[section]{placeins}
\usepackage{amsmath,amssymb,mathtools,amsthm,dsfont}

\DeclareFontFamily{U}{mathb}{\hyphenchar\font45}
\DeclareFontShape{U}{mathb}{m}{n}{
      <5> <6> <7> <8> <9> <10> gen * mathb
      <10.95> mathb10 <12> <14.4> <17.28> <20.74> <24.88> mathb12
      }{}
\DeclareSymbolFont{mathb}{U}{mathb}{m}{n}
\DeclareMathSymbol{\Asterisk}{3}{mathb}{"06}
\usepackage{mathrsfs}
\DeclareMathAlphabet{\mathscrbf}{OMS}{mdugm}{b}{n} 
\usepackage{easybmat} 
\usepackage{cleveref}

\DeclareMathOperator{\wt}{wt}
\newcommand{\wtp}{\wt_{\mathsf P}} 
\DeclareMathOperator{\conv}{\ast}
\DeclareMathOperator{\Conv}{\Asterisk}

\newcommand{\su}[1]{^{(#1)}}
\newcommand{\primed}{^{\prime}}
\newcommand{\iverson}[1]{1[#1]} 

\newcommand{\fourier}[1]{\widetilde{#1}}

\newcommand{\stabs}{\mathscr{S}}
\newcommand{\stabssymp}{\mathscrbf{S}}
\newcommand{\stabsclass}{C^\perp}

\newcommand{\inprod}[2]{\langle #1, #2 \rangle}
\newcommand{\inprodp}[2]{\inprod{#1}{#2}_{\mathsf P}}
\newcommand{\inprodds}[2]{\inprod{#1}{#2}_{\mathrm {DS}}}

\newcommand{\syn}{\mathcal{O}} 

\newcommand{\Chanset}{\Gamma}
\newcommand{\chanset}{\gamma}

\newcommand{\detectableset}{\Chanset\su D}

\newcommand{\allones}[1]{\mathbf{1}_{#1}^{}\mathbf{1}^T_{#1}}

\newcommand{\abs}[1]{\left\vert #1 \right\vert} 
 
\newcommand{\absb}[1]{\bigl\vert #1 \bigr\vert}

\newtheorem{theorem}{Theorem} 
\newtheorem{lemma}[theorem]{Lemma} 
\newtheorem{corollary}[theorem]{Corollary} 

\makeatletter 
 \hypersetup
 {pdftitle = {Pauli channels can be estimated from syndrome measurements in quantum error correction},
       pdfauthor = {Thomas Wagner, Hermann Kampermann, Dagmar Bruss, Martin Kliesch},
       pdfsubject = {Quantum computing},
       pdfkeywords = {
                      quantum, error, correction, QEC, code, codes, QECC, 
                      data, syndrome, 
                      stabilizer, decoding, 
                      pure, distance, 
                      estimation, noise, from, syndrome, statistics, 
                      probabilistic modelling, estimation, inverse problem, 
                      measurements, identifiability, 
             	    Pauli, group, phase space, symplectic, 
                      logical error, method of moments, moments, 
                      Schur complement, iterated, 
                      binomial, system, equation, 
                      Boolean, Fourier, transform, 
                      convolutional, factor graph, 
                      orthogonal array, 
                      Möbius, inclusion, exclusion, principle, 
                      distributed source coding
                      } 
  }
\makeatother

\begin{document}
\title{Pauli channels can be estimated from syndrome measurements in quantum error correction
}

\author{Thomas Wagner}
\email[]{thomas.wagner@uni-duesseldorf.de}
\author{Hermann Kampermann}
\author{Dagmar Bru\ss}
\author{Martin Kliesch}
\email[]{science@mkliesch.eu}
\affiliation{Institut für Theoretische Physik, Heinrich-Heine-University D{\"u}sseldorf,  Germany}
\begin{abstract}
\end{abstract}

\maketitle

\begin{abstract}
The performance of quantum error correction can be significantly improved if detailed information about the noise is available, allowing to optimize both codes and decoders.
It has been proposed to estimate error rates from the syndrome measurements done anyway during quantum error correction.
While these measurements preserve the encoded quantum state, it is currently not clear how much information about the noise can be extracted in this way.
So far, apart from the limit of vanishing error rates, rigorous results have only been established for some specific codes.

In this work, we rigorously resolve the question for arbitrary stabilizer codes.
The main result is that a stabilizer code can be used to estimate Pauli channels with correlations across a number of qubits given by the pure distance.
This result does not rely on the limit of vanishing error rates, and applies even if high weight errors occur frequently. 
Moreover, it also allows for measurement errors within the framework of quantum data-syndrome codes.
Our proof combines Boolean Fourier analysis, combinatorics and elementary algebraic geometry.
It is our hope that this work opens up interesting applications, such as the online adaptation of a decoder to time-varying noise.
\end{abstract}

\section{Introduction}
Quantum error correction is an essential part of most quantum computing schemes.
It can be significantly improved if detailed knowledge about the noise affecting a device is available, as it is possible to optimize codes for specific noise \cite{robertson2017_tailoredcodessmallmemories, florjanczykbrun_insituadaptiveencoding}. A prominent example is the XZZX-surface code \cite{ataides2021_XZZXSurfaceCode}.
Furthermore, common decoding algorithms, such as minimum weight matching \cite{higgott2021_pymatching, dennis_topologicalquantummemory, nickersonbrown_analyisngcorrelatednoisesrfacecodeadaptvedecoding, obrien_adaptiveweightestimator} and belief propagation (see e.g.\ \cite{babar_ldpcreview}), can incorporate information about error rates to return more accurate corrections.
Other examples of decoders that can incorporate information about error rates are weighted union find \cite{huang2020_weightedunionfind} and tensor network decoders \cite{chub2022_tensornetworkdecoder, darmawan2018_tensornetworkdecoder}. 
The latter can also deal with correlated noise models, but are relatively slow.
In the context of stabilizer codes, noise is usually modeled using Pauli channels, which are simple to understand and simulate.
However, Pauli noise is more than a mere toy model, since randomized compiling can be used to project general noise onto a Pauli channel \cite{wallman_randomizedcompiling}, 
which has also been demonstrated experimentally \cite{ware2021_randomizedcompiling-experimental}.
Furthermore, it is known that quantum error correction decoheres noise on the logical level \cite{beale2018_qecdecoheresnoise}.
Consequently, there has been much interest and progress in the estimation of Pauli channels \cite{flammia2021_paulilearningpopulationrecovery,
harper2020_sparsepauliestimation,
flammia2020_efficientestimationofpaulichannels,
harper2020_efficientlearningofquantumnoise}.
Complementary to the standard benchmarking approaches, it has been suggested to perform (online) estimation of channels just from the syndromes of a quantum error correction code itself \cite{fujiwara_instantaneouschannelestimationcss,
fowler_scalableextractionoferrormodelsfromqec,
huo_2017,
wootton_qiskitbenchmarking,
florjanczykbrun_insituadaptiveencoding,
obrien_adaptiveweightestimator,
combes_insitucharofdevice, wagner2021_optimalnoiseestimationfromsyndromes}.
Such a scheme uses only measurements that do not destroy the logical information.
It is thus suited for online adaptation of a decoder to varying noise, for example by adapting weights in a minimum weight matching decoder \cite{obrien_adaptiveweightestimator, huo_2017}. It furthermore results in a noise model that can be directly used by the decoder \cite{fowler_scalableextractionoferrormodelsfromqec}.  
Experimentally, online optimization of control parameters in a 9-qubit superconducting quantum processor has been demonstrated in an experiment by Google \cite{Kelly_ScalableInSituCalibrationDuringErrorDetection}.

Since the state of the logical qubit is not measured, some assumptions are necessary in order for this estimation to be feasible.
However, the precise nature of these assumptions is currently not well understood.
In the general case, it is unclear for which combinations of noise models and codes the parameters can be identified using only the syndrome information. 
Apart from heuristics in the limit of very low error rates \cite{fowler_scalableextractionoferrormodelsfromqec, wootton_qiskitbenchmarking, huo_2017}, only two special cases have been rigorously treated.
For codes admitting a minimum weight matching decoder, such as the toric code, identifiability of a circuit noise model was shown by \citet{obrien_adaptiveweightestimator}.
In \cite{wagner2021_optimalnoiseestimationfromsyndromes}, identifiability results for the restricted class of perfect codes were proven. 

It is not a priori clear that an estimation of the error rates just from the syndrome statistics should be possible at all for arbitrary stabilizer codes.
There are several objections one could raise. For one, as mentioned above, the state of the logical qubit is not measured, so there is only limited information contained in the syndromes.
Phrased another way, there are generally exponentially many errors with the same syndrome, which we therefore cannot distinguish.
This also implies that the probability of a syndrome is an exponentially large sum of different error rates, and solving such a system for the error rates appears difficult at best.
While it has been suggested to simplify the problem by only taking into account the lowest weight error compatible with each syndrome \cite{fowler_scalableextractionoferrormodelsfromqec, wootton_qiskitbenchmarking, huo_2017}, this approximation strategy leads to demonstrably sub-optimal estimators \cite{wagner2021_optimalnoiseestimationfromsyndromes}.

\section{Results}

In this work, we show that the estimation task can be solved for arbitrary stabilizer codes.
For any given quantum code, we describe a general class of Pauli channels whose parameters can be identified from the corresponding syndrome statistics. 
Our results also take into account measurement errors by using the framework of quantum data-syndrome codes \cite{ashikhmin_quantumdatasyndromecodes,fujiwara_datasyndromecodes,delfosse_beyondsingleshotfaulttolerance}.
We prove that a large amount of information can be extracted from the syndrome statistics. 
In the following theorem we make this statement more precise in terms of the pure distance, which is defined as the minimum weight of an undetectable error (see \Cref{sec:quantum} for details).
The pure distance measures up to which weight errors can necessarily be distinguished by their syndrome, and for most codes it coincides with the weight of the smallest stabilizer.
Since there can be undetectable but logically trivial errors, the pure distance is usually much smaller than the distance.
As an example, for the family of toric codes the pure distance is 4 independent of code size. 
More generally, the pure distance will be constant for any family of quantum low density parity check (LDPC) codes, since the stabilizer weights are constant.
 
We are interested in estimating a Pauli channel, i.e.\ a probability distribution over Pauli errors, which describes the new error occurring before each round of error correction (working in a phenomenological noise model). 
It is commonly assumed that errors on each qubit occur independently, in which case the channel is described by one error rate for each qubit.
In this case, we say the noise is uncorrelated.
In a more general setting, the Pauli channel could act on many qubits, such that errors on these qubits are not independent.
If e.g.\ 2 qubit errors occur that are not a combination of independent single qubit errors, we say that the noise is correlated over 2 qubits. 
The corresponding 2 qubit error rates are then not simply a product of the single qubit error rates and must be additionally specified. 
For example, if $P(X_1X_2) \neq P(X_1)P(X_2)$, the errors on the first two qubits are correlated. 
This notion of correlations is made precise in \Cref{sec:MomentsAndNoise} for classical and \Cref{sec:Quantum-Identifiability} for quantum codes.
We now give an informal statement of our main results.
A formal statement is given in \Cref{thm:main-full}, which also takes into account many detectable errors beyond the pure distance.

\begin{theorem}[Main result, informal]
A stabilizer code with pure distance $d_p$ can be used to estimate Pauli noise with correlations across up to $\lfloor \frac{d_p-1}{2} \rfloor$ qubits.
\label{thm:main-informal}
\end{theorem}

While this result is stated in a non-constructive fashion, its derivation suggests a concrete estimation protocol. 
We give a first heuristic discussion of the resulting estimators in \Cref{sec:PracticalEstimation}. 
A detailed analysis of such estimators is ongoing work. 

The key idea behind \Cref{thm:main-informal} is that the error distribution is fully described by a set of moments.
We will show that these moments can be estimated up to their sign by solving a polynomial system of equations.
The appearance of multiple discrete solutions, differing in some signs, reflects the symmetries of the problem.
For example in the case of single qubit noise, there is one symmetry for each logical operator of the code.
However, under the additional mild assumption that all error rates are smaller than $\frac 12$, all moments must be positive, and thus the estimate is unique. 
We stress that our result does not rely on the limit of vanishing rates, and still applies in the presence of high weight errors. Perhaps surprisingly, even different errors with the same syndrome can occur frequently, as long as they arise as a combination of independent lower weight errors. Our result is arguably the strongest one can reasonably expect, since error rates cannot be identifiable if multiple independently occurring errors have the same syndrome.

The \emph{adaptive weight estimator} by \citet{obrien_adaptiveweightestimator} can be viewed as solving a special instance of the general equation system we present here, which is applicable only for codes that admit a minimum weight matching decoder. This  connection is explained in \Cref{sec:ToricCodeExample}.

Our arguments combine Boolean Fourier analysis, combinatorics and elementary algebraic geometry. Interestingly, a connection between Pauli channel learning and Boolean Fourier analysis was recently pointed out in an independent work by \citet{flammia2021_paulilearningpopulationrecovery}.

We start our discussion with the motivating example of the toric code with independent Pauli-$X$ errors, where our results take a particularly simple form.
We then discuss the general results.
We first discuss the setting of classical codes, and later extend the results to quantum codes.
This is for ease of exposition, since the presentation is considerably simplified in the classical setting and thus the underlying concepts become more apparent.
 We stress that, in contrast to the quantum setting, classically one can measure the individual bits of a codeword without destroying the encoded information.
Thus, in the classical setting one does not have to rely solely on the syndrome information and can use easier techniques for error rates estimation.
However, for a quantum code this is indeed the only information that can be measured without destroying the encoded state, and this is where our results are most relevant.
Classically, our approach might still be useful in the setting of distributed source coding \cite{zia_ldpcerrorrateestimation}.

\subsection{Example: the toric code}
\label{sec:ToricCodeExample}

As a motivating example for the methods and proofs in the following sections, we derive an estimator for the simple setting of completely uncorrelated bit-flip noise on the toric code.
That is, we assume that errors on different qubits occur independently, possibly with a different rate for each qubit, and that only bit-flip (Pauli-$X$) errors occur on each qubit.
Thus, we consider the toric code essentially as a classical code.
This constitutes a simple alternative derivation of the solution given by \citet{obrien_adaptiveweightestimator}.

\begin{figure}[t]
\centering
\includegraphics[scale=2]{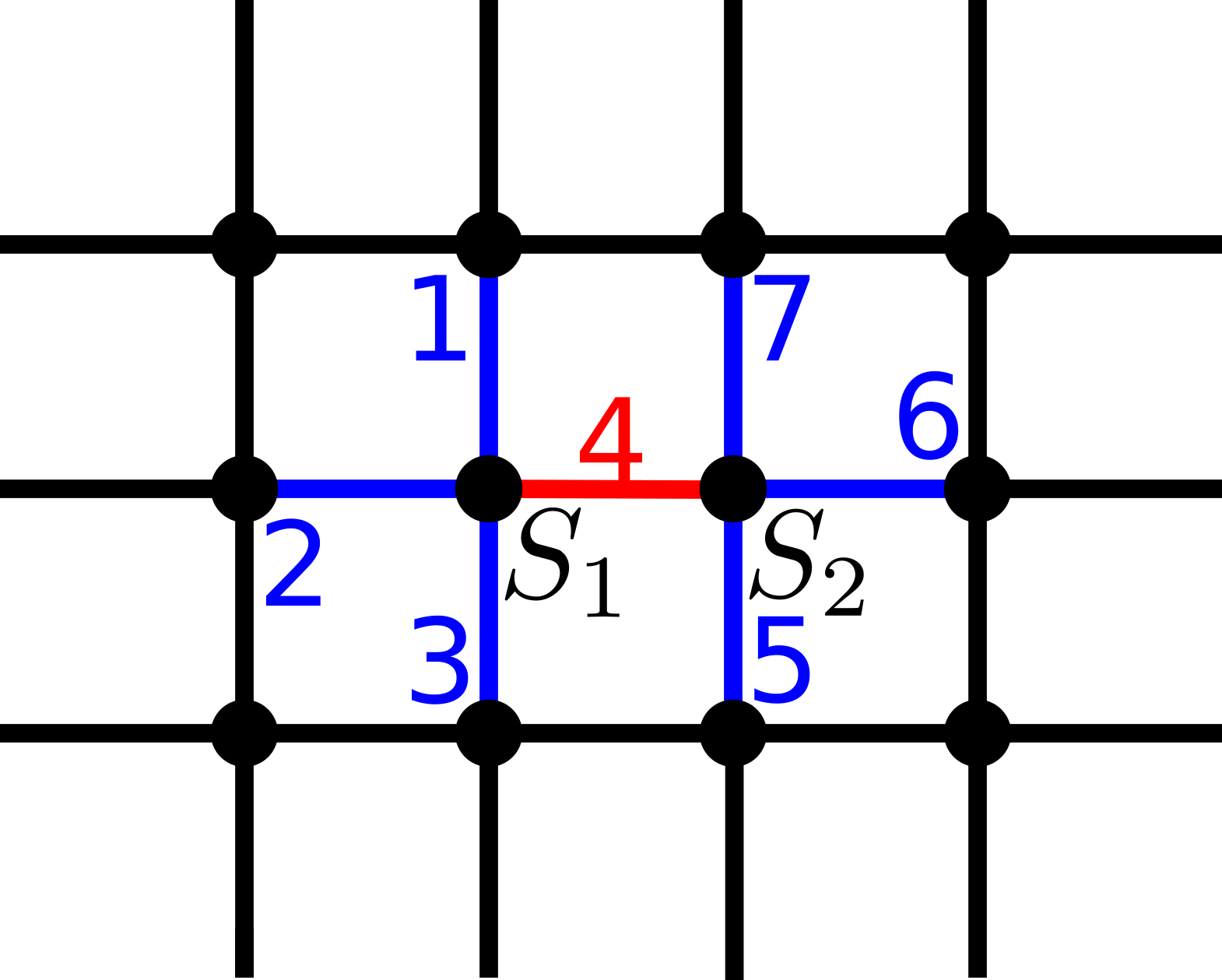}
\caption{A qubit (edge) and two adjacent star operators (vertices) on the toric code.}\label{fig:ToricCode}
\end{figure}

We focus on a single qubit and two adjacent $Z$-stabilizers, as illustrated in \cref{fig:ToricCode}.
Our task is to estimate the error rate $p_4$ of the marked qubit, or equivalently, the expectation value $E(Z_4) = 1 - 2p_4$.
Since errors on each qubit are independent, the expectation of a stabilizer measurement is simply the product of expectations of the adjacent bits.
Therefore we obtain the following three equations,
\begin{align*}
&E(S_1) = E(Z_1)E(Z_2)E(Z_3)E(Z_4) \\
&E(S_2) = E(Z_4)E(Z_5)E(Z_6)E(Z_7)  \\
&E(S_1S_2) = E(Z_1)E(Z_2)E(Z_3)E(Z_5)E(Z_6)E(Z_7)\, .
\end{align*}
This system admits a straightforward solution for the expectation of $Z_4$, 
\begin{equation}
E(Z_4) = \pm \sqrt{\frac{E(S_1)E(S_2)}{E(S_1S_2)}} \, .
\label{eq:SolutionToric}
\end{equation}
This coincides with the solution given by \citet[eq. (14)]{obrien_adaptiveweightestimator}, as explained in \Cref{sec:AppendixSO-Solution}. 
Notice that there is a choice of sign, which corresponds to deciding whether $p_4 > \frac{1}{2}$ or $p_4 < \frac{1}{2}$.
However, under the assumption $p_4 < \frac{1}{2}$ the solution is unique.
\\

\subsection{Identifiablity for classical codes}
Let us now turn to the general setting, first for classical codes.
We are interested in whether an error distribution can be estimated from repeated syndrome measurements alone. 
This is certainly not possible for completely arbitrary noise, since we do not measure the state of the logical qubit. However, we will show that for Pauli (and measurement) noise with limited correlations, the estimation is possible. 
As mentioned above, we will start with the setting of classical codes and perfect syndrome measurements. Then, we will show how to extend these results to quantum (data-syndrome) codes.

The key insight underlying our proof is that the estimation problem is best phrased in terms of moments instead of error rates. The proof then proceeds in the following steps. First, we notice that Fourier coefficients of the error distribution correspond to moments, and see that some of these moments can be estimated from the syndrome statistics. Then, we show that under certain independence assumptions, the full error distribution is characterized completely by a set of low weight transformed moments. We find that these transformed moments are related to the measured moments via a polynomial equation system. This system is described by a coefficient matrix $D$ whose rows essentially correspond to elements of the dual code. We then use local randomness properties of the dual code to find an explicit expression for the symmetric squared coefficient matrix $D^TD$, and finally show that $D^TD$ has full rank by using iterated Schur complements. 
This implies that the equation system has discrete solutions.

\subsubsection{Notation}
We use the short-hand notation $[n]\coloneqq \{1, \dots,n\}$ for the set of the first $n$ natural numbers.
For any set $A$, we denote its powerset as $2^A \coloneqq \{ B \mid B \subseteq A  \}$.
The field with two elements is denoted by $\mathbb F_2$. 
Often we will use $\mathbb F_2^n$ as a vector space over that field. 
That is, for $a,b\in \mathbb F_2^n$ the sum $a+b$ is understood to be taken component-wise modulo $2$. 
All vectors are to be understood as column vectors.
By $\wt(a) \coloneqq |\{i: a_i \neq 0\}|$ we denote the \emph{weight} of $a$. 
For any logical statement $X$ we denote by $\iverson{X}$ the \emph{Iverson bracket} of $X$, which assumes the value $1$ if $X$ is true and $0$ otherwise.
Naturally, trivial products, i.e.\ products over the empty set, are set to $1$ as in 
$\prod_{x\in \emptyset} f(x) \coloneqq 1$.
By $I_k$ we denote the $k\times k$ identity matrix and suppress $k$ when it can be inferred from the context.

\subsubsection{Classical codes and Boolean Fourier analysis}
Let us start with some basic elements of Boolean Fourier analysis. A detailed review of the topic is given in \cite{odonnel2014_analysisofbooleanfunctions} (note that we use a different normalization convention here). 
We frequently identify a vector $ s \in \mathbb{F}_2^n$ with its \emph{indicator set} $\{ i \in [n] \;  : \; s_i \neq 0\}$. 
For example, for $ s = (0,1,0,1,1,0) \in \mathbb{F}_2^6$ we also write
\begin{equation}
 s =  \{2,4,5\} \subseteq [6] \, ,
\end{equation}
such that we can write $4 \in  s$. 
For each $ s \in \mathbb{F}_2^n$, we define the \emph{parity function}
\begin{equation}
\begin{split}
&\chi_{ s}: \; \mathbb{F}_2^n \rightarrow \{-1,+1\}\, , \nonumber \\
&\chi_{ s}(e) = (-1)^{ s \cdot e} = (-1)^{\sum_{i \in  s} e_i} \, ,
\end{split}
\end{equation}
which is the group character of the abelian group $(\mathbb{F}_2^n,+)$.
For a function $f: \mathbb{F}_2^n \rightarrow \mathbb{R}$, its \emph{Boolean Fourier transform} $\fourier f$ is the function
\begin{equation}
\fourier f: 2^{[n]} \rightarrow \mathbb{R} \, , \quad
\fourier f( s) = \sum_{ e \in \mathbb{F}_2^n} f( e)\chi_{ s}( e) \; .
\end{equation} 
The Boolean Fourier transform is also  known under the name Walsh-Hadamard transformation. 
However, there are different conventions which lead to transforms with different orderings of bit strings and different notations for the characters $\chi_s$.
This transformation is invertible, and the inverse is given by
\begin{equation}\label{eq:FTinv}
f( e) = \frac{1}{2^n}\sum_{ s \subseteq [n]} \fourier f( s)\chi_{ s}( e) \, .
\end{equation} 
For two functions $f,g: \mathbb{F}_2^n \rightarrow \mathbb{R}$, their \emph{Boolean convolution} $f \ast g$ is defined by
\begin{equation}
f \ast g ( e) = \sum_{{e}\primed \in \mathbb{F}_2^n} f( e)g( e +  e\primed) \, .
\label{eq:def:convolution}
\end{equation}
As expected, convolutions become products under Boolean Fourier transform, 
\begin{equation}
\fourier{f \conv  g} = \fourier{f}\, \fourier{g}\, .
\end{equation} 

A classical linear code $C$, encoding $k$ bits into $n$ bits, is a $k$-dimensional subspace of $\mathbb{F}_2^n$. It can be described by its parity check matrix $ H \in \mathbb{F}_2^{(n-k) \times n}$, whose rows span the dual code $\stabsclass = \{ a \in \mathbb{F}_2^n \; : \;  a \cdot  c = 0 \; \forall  c \in C \}$. $\stabsclass$ can alternatively be interpreted as the set of parity functions which are 1 on all codewords. When an error $ e \in \mathbb{F}_2^n$ occurs, the outcomes of the parity measurements $ H$ can be summarized by the \emph{syndrome} 
\begin{equation}\label{eq:def:syn}
       \syn( e) \coloneqq  H  e \, .
 \end{equation} 
From $\syn(e)$, we can calculate the value of any parity measurement $ s \in \stabsclass$. We will assume that repeated rounds of syndrome measurements and corrections are performed. In each round, an error occurs according to the \emph{error distribution} $P: \mathbb{F}_2^n \rightarrow [0,1]$. We call the corresponding distribution of syndromes the \emph{syndrome statistics}. The error correction capabilities of a classical linear code are indicated by its \emph{distance}, which is defined as the smallest weight of an undetectable error, which is the same as the smallest weight of a codeword, 
\begin{equation}
d \coloneqq \min \{\wt( e) \; : \;  e \in C \setminus \{0\} \} \, .
\end{equation}

\subsubsection{Moments and noise model}
\label{sec:MomentsAndNoise}
We denote the Fourier transform of the error distribution $P: \mathbb{F}_2^n \rightarrow [0,1]$ with $E$, i.e.\ for each $ a \subset [n]$ 
\begin{equation}\label{eq:Edef}
E( a) \coloneqq \sum_{e \in \mathbb{F}_2^n} (-1)^{ a \cdot e} P(e) = \fourier P( a) \, .
\end{equation}
We interpret this as a \emph{moment} of the distribution $P$, since $E( a)$ is exactly the expectation value of the parity measurement $ a$ if one were to measure it repeatedly on errors distributed according to $P$. In particular, for $ s \in \stabsclass$, the corresponding moment $E( s)$ can be computed from the measured syndrome statistics, i.e. from the empirically measured frequency with which each syndrome occurs in repeated rounds of error correction.
We will always assume $E( a) \neq 0$ for all $ a \subseteq [n]$, which is, for example, fulfilled if $P(0) > 0.5$.
Thus, our task is to find the distribution $P$, given some of its Fourier coefficients.

We will show that this task is feasible if there is some independence between errors on different subsets of bits, such that there are no correlations across a large number of qubits at once.
To formalize this idea, we introduce a set of \emph{channel supports} $\Chanset \subseteq 2^{[n]}$ (we will consider a concrete choice of $\Chanset$ later).
For each $\chanset \in \Chanset$, there is an error channel that acts only on this subset, i.e.\ its error distribution $P_{\chanset}: \mathbb{F}_2^n \rightarrow [0,1]$ is only supported on errors $ e \in \mathbb{F}_2^n$ that are $0$ outside of $\chanset$.
Equivalently, identifying vectors with their indicator set, we can view this as  a distribution $P_\chanset : 2^{[n]} \rightarrow [0,1]$ which is only supported on $2^\chanset$.
Since all the individual channels act independently, and the total error is the sum of the individual errors, the total distribution of errors is then given by 
\begin{equation}\label{eq:conv}
P = \underset{\chanset \in \Chanset}{\Conv} P_\chanset \, ,
\end{equation}
with the convolution \eqref{eq:def:convolution}.
Physically, this expresses the action of several independent sources of errors, each only acting on a limited number of qubits. If $\Chanset = \{ \{1\},\dots,\{n\} \}$, this reduces to the commonly used model where each bit is affected by an independent channel. 
As a note, mathematically, this model corresponds to a convolutional factor graph, as introduced by \citet{mao2005_convolutionalfactorgraphs}.

This representation in terms of individual channels is  an over-parametrization of the total distribution if some supports in $\Chanset$ overlap. A first idea for a non-redundant distribution would be to use the set of all moments, but these are not all independent. Therefore, we define transformed moments $F( a)$ via a (generalized) inclusion-exclusion computation. Heuristically, we want to divide out all correlations on proper subsets of $ a$ in order to capture only correlations across the full size of $ a$.
This bears some similarity to the canonical parametrization of Markov random fields, see \cite{koller_pgm}.
For $ b \subseteq  a \subseteq [n]$, the \emph{Möbius function} (on the partially ordered set $2^{[n]}$) is given by
\begin{equation}
\mu( a, b) = (-1)^{| a \setminus  b|} = (-1)^{| a|+| b|} \, .
\end{equation}
For each $ a \subseteq [n]$, we define a transformed moment $F(a)$ via the \emph{inclusion-exclusion transform} 
\begin{equation}
F( a) \coloneqq \prod_{ b \subseteq  a} E( b)^{\mu( a, b)} \, .
\label{eq:F-E}
\end{equation}
It follows from the (multiplicative) generalized inclusion-exclusion principle (e.g.\ \cite{Aigner2007_ACourseInEnumeration}[Theorem 5.3], \cite[Theorem A.2.4]{roman2006_fieldtheory}) that the inverse of this transformation is given by  
\begin{equation}\label{eq:MoebiusInv}
E( a) = \prod_{ b \subseteq  a} F( b) \, .
\end{equation}
Now we will show that, depending on the choice of $\Chanset$, a small subset of transformed moments is already sufficient to uniquely specify the distribution.
Remember $E_\chanset( a) = \fourier P_\chanset( a)$ from \eqref{eq:Edef}.
Since $P$ is given by the convolution \eqref{eq:conv} of all individual error distributions $P_\chanset$ we have,
\begin{equation}
E(a) = \prod_{\chanset \in \Chanset} E_\chanset( a \cap \chanset) \, ,
\label{eq:E-gamma}
\end{equation}
where we used that $P_\chanset$ is only supported on $2^\chanset$, and thus $ E_\chanset(a) = \tilde P_\chanset(a) = \tilde P_\chanset(a \cap \chanset) = E_\chanset(a \cap \chanset)$.

We can now express the transformed moments by the parameters of the individual channels.  The proof is given in \Cref{sec:Appendix-ProofOfLemma:RescaledMoments-1}.
\begin{lemma}
For $\Gamma\subseteq 2^{[n]}$ let $E: 2^{[n]}\to\mathbb R$ satisfy \eqref{eq:E-gamma}. 
Then, the inclusion-exclusion transform \eqref{eq:F-E} satisfies
\begin{equation}
F( a) = \prod_{\chanset \in \Chanset :  a \subseteq \chanset} \prod_{ b \subseteq  a}E_\chanset( b)^{\mu( a, b)} 
\end{equation}
for all $a\subseteq [n]$. 
In particular, $F( a) = 1$ if there is no $\chanset \in \Chanset$ such that $ a \subseteq \chanset$.
\label{lemma:RescaledMoments-1}
\end{lemma}

Notice that trivially $F(\emptyset) = 1$. Thus, in conclusion, it suffices to determine the transformed moments $F( a)$ for $ a$ in
\begin{equation}
\hat \Chanset \coloneqq \{\emptyset \neq  a \subseteq \chanset \; : \; \chanset \in \Chanset \} \, .
\label{eq:def:Gammahat-classical}
\end{equation}
The standard moments \eqref{eq:MoebiusInv} are then determined by
\begin{equation}\label{eq:EFviaGammaHat}
E( a) = \prod_{ b \in \hat \Chanset :  b \subseteq  a} F( b) \, .
\end{equation}
Finally, the error distribution $P$ is determined from the standard moments by applying the inverse Fourier transform \eqref{eq:FTinv} to \eqref{eq:Edef}.

\subsubsection{Identifiability and binomial systems}

As discussed in the previous section, learning the transformed moments $(F( a))_{ a \in \hat \Chanset}$ is sufficient to uniquely determine the error distribution. However, the only thing we can measure are the standard moments $(E( s))_{ s \in \stabsclass}$ corresponding to elements of the dual code, since these fully describe the syndrome statistics.
Thus, learning the error distribution boils down to determining the transformed moments $(F( a))_{ a \in \hat \Chanset}$ from this information.
In particular, the error distribution is identifiable from the syndrome statistics if and only if the following system of polynomial equations admits a unique solution,
\begin{align}
&\text{Given } 
(E( s))_{ s \in \stabsclass} 
\text{ find } 
(F( b))_{ b\in \hat \Chanset} 
\text{ satisfying }
\nonumber \\[.2em]
&E( s) = \prod_{ b \in \hat \Chanset :  b \subseteq  s} F( b)
\qquad 
\forall  s \in \stabsclass \setminus \{0\} \, , 
\label{eq:BinomialSystemE-F}
\end{align}
where we have omitted the trivial equation arising from the zero element of the dual code.
This is a \emph{binomial} system, i.e.\ each equation only has two terms, a constant on the left-hand side and a monomial on the right-hand side. 
In the notation of \cite{chen2014_solutionsofbinomials}, the binomial system \eqref{eq:BinomialSystemE-F} can be expressed by the \emph{coefficient matrix} $D$ whose rows are labeled by elements of $\stabsclass\setminus \{0\}$ and whose columns are labeled by elements of $\hat \Chanset$.
The entry $D[ s, a]$ is the exponent of $F( a)$ in the monomial on the right-hand side of equation \eqref{eq:BinomialSystemE-F} for $E( s)$. 
We have transposed the notation comparing to \cite{chen2014_solutionsofbinomials}.
Explicitly this means that the $\abs{\stabsclass\setminus \{0\}}\times \absb{\hat \Chanset}$-matrix $D$ has the elements
\begin{equation}\label{eq:def:D1}
D[ s, a] = \begin{cases}
		 1 & \text{ if }  a \subseteq  s \\
		 0 & \text{ otherwise }
		 \end{cases} \, , 
\end{equation}
and the system is given by 
$\vec E = \vec F^D$ \cite[Eq.~(3)]{chen2014_solutionsofbinomials}, where $\vec E$ is the vector containing the elements $E( s)$ and $\vec F$ the elements $F( a)$, as appearing in the system \eqref{eq:BinomialSystemE-F}.  
In the special case of single bit noise, i.e.\ $\Chanset = \hat \Chanset = \{ \{i\} : i \in [n]\}$, the rows of the coefficient matrix $D$ are exactly the dual codewords $ s \in \stabsclass$. 

For now, let us assume that $D$ has full rank. It then follows from the theory of binomial system solving \cite[Proposition~2]{chen2014_solutionsofbinomials}
that the system \eqref{eq:BinomialSystemE-F} has a finite number of solutions, and these solutions only differ by multiplying some parameters with complex roots of unity.
Since we are only interested in real solutions, this means we can determine the transformed moments $F( a)$, and thus also the standard moments $E( a)$, up to a sign. 
Thus, if we restrict all moments to be positive, the error distribution is uniquely determined by the syndromes. A simple condition for all moments to be positive is that the error probabilities of all channels are smaller than $\frac 12$, i.e.\ $P_\chanset(0) > \frac 12$ for all $\chanset \in \Chanset$. This implies $E_\chanset( a) > 0$ and thus $E( a) > 0$ for all $ a \subseteq [n]$. In conclusion, the error distribution can be estimated uniquely from the syndromes, assuming that $D$ has full rank and that the error probability of each channel is smaller than $\frac 12$.

Let us stress that the appearance of multiple solutions, differing by the signs of some moments, reflects actual symmetries of the syndrome statistics as a function of the error rates.
For example, consider again case of independent single bit errors, such that $\Chanset = \hat \Chanset = \{ \{i\} : i \in [n]\}$.
Then the transformed moments $F(\{i\})$ are equal to the standard moments $E(\{i\})$, and the rows of $D$ are simply the elements of the dual code. 
Thus, each row of $D$ has even overlap with every codeword $c \in C$.
Thus, by \eqref{eq:BinomialSystemE-F}, flipping the signs of all moments $(E(\{i\}))_{i \in c}$ on the support of the codeword $c$ does not change the measured moments $(E( s))_{ s \in \stabsclass}$.
Since $E(\{i\}) = 1-2p_i$, where $p_i$ is the rate of errors on the $i$-th qubit, this simply means that flipping the error rates around $\frac 12$ on a codeword does not affect the syndrome statistics, i.e.\ we cannot distinguish these two sets of error rates. This observation shows that each codeword corresponds to a symmetry of the identifiability problem.

\subsubsection{The rank of the coefficient matrix}
We will now establish the most general set of error distributions for which unique identification of the error rates is possible. By the discussion in the previous section, this corresponds to finding the most general choice of $\Chanset$ for which the system \eqref{eq:BinomialSystemE-F} is solvable, i.e.\ for which the coefficient matrix $D$ is of full rank.

Denote the sets of bits that only support \emph{detectable} errors as
\begin{equation}\label{eq:def:Delta}
\detectableset \coloneqq \{\chanset \subseteq [n] \; : \; \syn(e) \neq 0 \; \forall e \subseteq \chanset \} \, , 
\end{equation}
where $\syn(e)$ is the syndrome \eqref{eq:def:syn} of $e$.
It is clear that, if one of the channels $\chanset \in \Chanset$ supports an undetectable error, one cannot estimate the corresponding error rate. 
Thus, identifiability can only hold if $\chanset \in \detectableset$. 
Similarly, if two different channels $\chanset_1,\chanset_2$ contain two errors with the same syndrome, the syndrome statistics only depends on the combined rate of those errors and thus the rates are not identifiable. Identifiability of the noise channel can only hold if $\chanset_1 \cup \chanset_2 \in \detectableset$ for all $\chanset_1,\chanset_2 \in \Chanset$. We will show that these are in fact the only restriction that must be fulfilled.

\begin{theorem}[Classical identifiability condition]
Consider a classical code with $n$ bits subject to noise with an error distribution described by channel supports $\Chanset \subseteq 2^{[n]}$.
Then the coefficient matrix $D$ defined in \eqref{eq:def:D1} is of full rank if and only if $\chanset_1 \cup \chanset_2 \in \detectableset$ for all $\chanset_1,\chanset_2 \in \Chanset$.
\label{thm:main}
\end{theorem}

The proof is given in Sections \ref{sec:OrthogonalArry-maintext}. 
Even a combination of two channels must not support an undetectable error.
We can also give an equivalent condition in terms of $\hat \Chanset$, which represents the set of errors that can occur ``independently''.
As explained in \Cref{sec:Appendix-ChannelSupports},
the assumption $\chanset_1 \cup \chanset_2 \in \detectableset$ for all $\chanset_1,\chanset_2 \in \Chanset$ is equivalent to $\syn(e_1) \neq  0$ and $\syn(e_1 + e_2) \neq 0$ for all $e_1,e_2 \in \hat \Chanset$ (viewed as binary vectors).
In other words, the error distribution is identifiable if undetectable errors and errors that have the same syndrome only occur as a combination of independent errors.
We stress that this is substantially weaker than the assumption that errors with the same syndrome never occur (which was made in some previous works such as \cite{fowler_scalableextractionoferrormodelsfromqec, huo_2017}).
Indeed, in the total error distribution different errors with the same syndrome can occur frequently, and we are still able to identify the error rates. We only assume that such errors arise as a combination of independent errors.

The most important example is the following. Consider an error model where there is an independent channel on every subset of $t$ bits, i.e.\ choose $\Chanset$ as 

\begin{align}
&\Chanset = \Chanset_t \coloneqq \{\chanset \subseteq [n] \; : \; |\chanset| = t \} \\
&\hat \Chanset = \Chanset_{\leq t} \coloneqq \{e \subseteq [n] \; : \; 0 < |e| \leq t \} \, . \label{eq:hatGamma_t}
\end{align}

This means we only have correlations across at most $t$ bits.
Then, \Cref{thm:main} implies that the error distribution is identifiable as long as $d \geq 2t+1$. 
\begin{corollary}[Distance based identifiability condition]
If the noise is described by channel supports $\Chanset = \Chanset_t$, the coefficient matrix $D$ is of full rank if the code has distance $d \geq 2t+1$.
\end{corollary}
\begin{proof}
For a code with distance $d$, $\detectableset$ contains every set of size at most $d-1$.
\end{proof}

For $t=1$, we see that error rates of the standard single bit noise model can be identified as long as $d \geq 3$. Informally, identification is possible if error correction is possible. 

\subsubsection{Orthogonal array properties of the coefficient matrix}
\label{sec:OrthogonalArry-maintext}
This section is devoted to the proof of \Cref{thm:main}. Readers not interested in the proofs and mathematical techniques can skip to \Cref{sec:quantum} for the results in the quantum setting.
The first part of the proof is based on local randomness properties of the dual code. The elements of the dual code form a so-called \emph{orthogonal array}. Since the entries of the coefficient matrix $D$ are related to the dual code, we can use this property to derive an explicit expression for the symmetric squared coefficient matrix $D^TD$. The proof is then finished by computing the rank of $D^TD$, which is possible with the help of combinatorial results derived in the next section.
 
By $ T$ we denote the $|\stabsclass|\times n$-matrix formed by all elements in the span of the rows of the parity check matrix $ H$, i.e.\ the rows of $ T$ are exactly the elements of $\stabsclass$. It is known that the rows of $ T$ look ``locally uniformly random'' on any subset of up to $d - 1$ bits. One says that $ T$ is an orthogonal array of strength $d-1$. This property is relatively easy to see for linear codes (e.g.\ \cite{hedayat2012_orthogonalarrays}) and was shown for general (non-linear) classical codes by Delsarte \cite{delsarte1973_fundamentalcodecombinatorics}. 
We use a slightly extended version of the result for linear codes using the set $\detectableset$ from \eqref{eq:def:Delta} 
instead of the distance. A proof is given in \Cref{sec:Appendix-OrthogonalArrays}.
\begin{lemma}
Let $\chanset \in \detectableset$. In the restriction $ T^{|\chanset}$ of $ T$ to columns in $\chanset$, every bit-string appears equally often as a row.
\label{lemma:RandomStab-Classical}
\end{lemma}

In other words, the rows of $ T$ look locally uniformly random on any choice of bits that only supports detectable errors.

In \Cref{sec:RankM} we prove the following statement, as a corollary of \Cref{lemma:Schur-Complement-M-Simple}.
\begin{lemma}[Positive-definiteness of intersection matrix]
\label{lemma:MPosDef}
Let $M_t$ be the matrix whose rows and columns are labeled by the elements of $\Chanset_{\leq t}$ (from \eqref{eq:hatGamma_t}) and which is defined entry-wise by
\begin{equation}
\label{eq:def:M}
M_{t}[ a, b] = 2^{| a \cap  b|}  \, .
\end{equation}
Then $M_t$ is positive-definite.
\end{lemma}
We call the matrix $M_t$ the \emph{intersection matrix}.
The dimensions of $M_t$ depend on $n$, but we do not make this dependence explicit in our notation.
Using \Cref{lemma:MPosDef}, we can finish the proof of \Cref{thm:main}.
\begin{proof}[Proof of \Cref{thm:main}]
We will prove that the coefficient matrix $D$ from \eqref{eq:def:D1} has full rank (over $\mathbb{R}$) by proving that $D^TD$ has full rank. 
First, note that by the definition \eqref{eq:def:Delta} of $\detectableset$, for any $\chanset \in \detectableset$, all subsets $a \subseteq \chanset$ are also elements of $\detectableset$.
Thus, the assumption $\chanset_1 \cup \chanset_2 \in \detectableset$ from \Cref{thm:main} implies $ a \cup  b \in \detectableset$ for all $ a \cup  b \in \hat \Chanset$, by the definition \eqref{eq:def:Gammahat-classical} of $\hat \Chanset$. 
The dot product of the columns of $D$ labeled by $ a,  b \in \hat \Chanset$ (as binary vectors) is the number of elements $ s \in \stabsclass \setminus \{0\}$ such that $ a \cup  b \subseteq  s$. Thus, by \Cref{lemma:RandomStab-Classical}, this number is $|\stabsclass|\, 2^{-| a\cup  b|}$, i.e.\ $D^TD[ a, b] = |\stabsclass|\, 2^{-| a\cup  b|}$. 
We can write this as $D^TD[ a, b] = |\stabsclass|\, 2^{-| a|-| b|+| a \cap  b|}$. 
Re-scaling the rows and columns leads to the modified matrix $(D^TD)\primed[ a, b] = 2^{| a \cap  b|}$. 
Let us denote $t = \max \{| a| \; : \;  a \in \hat \Chanset \}$. 
Then $(D^TD)\primed$ is a principal sub-matrix of the intersection matrix $M_t$ defined in \eqref{eq:def:M}.
By \Cref{lemma:MPosDef}, $M_t$ is positive-definite. 
As a principal sub-matrix of a positive-definite matrix, $(D^TD)\primed$ is then also positive-definite and, in particular, has full rank. This implies that $D$ has full rank, which finishes the proof of \Cref{thm:main}. 
\end{proof}

\subsection{Extension to the quantum case}
\label{sec:quantum}
Now will consider the quantum setting, starting with a short overview of stabilizer codes.
We also explain the concept of quantum data-syndrome codes \cite{ashikhmin_quantumdatasyndromecodes,
fujiwara_datasyndromecodes,delfosse_beyondsingleshotfaulttolerance}, following \citet{ashikhmin_quantumdatasyndromecodes}, which allow for a unified treatment of data and measurement errors.
Then, we state and explain our main result \Cref{thm:main-full}, which is the formal version of \Cref{thm:main-informal}.
Finally, we explain how to prove this theorem by extending our arguments from the classical to the quantum case.

\subsubsection{Preliminaries}
\label{sec:quantum-preliminaries}
The Pauli group $\mathcal{P}^n$ on $n$ qubits is the group of \emph{Pauli strings} generated by the Pauli operators $\{I, X,Y,Z,\}$ with phases, 
\begin{equation*}
\mathcal{P}^n 
= \Bigl\{\epsilon \bigotimes_{i = 1}^n P_i \mid \epsilon \in \{\pm 1, \pm i\},\; P_{i} \in \{I,X,Y,Z\} \Bigr\} \, .
\end{equation*}
The \emph{weight} $\wt(P)$ of a Pauli string $P=\epsilon \bigotimes_{i = 1}^n P_i$ is the number of non-identity components $P_i$. 
Modding out phases, one obtains the \emph{effective Pauli group}
\begin{equation}
\mathsf P^n = \mathcal{P}^n / \{\pm 1, \pm i\} \, .
\end{equation}
We define the \emph{symplectic inner product} $\inprodp{\cdot}{\cdot}$ on $\mathsf P^n$ by
\begin{equation}
\inprodp{e}{e\primed} = 
\begin{cases} 
      1, & e \text{ and } e\primed \text{ anti-commute in }\mathcal P^n
      \\ 
      0, & e \text{ and } e\primed \text{ commute in }\mathcal P^n 
\end{cases} \, ;
\label{eq:def:SympProduct-Paulis}
\end{equation}
note that this expression is well-defined since the commutation relation does not depend on the choice of representatives in $\mathcal{P}^n$.
We identify $\mathsf P^1$ with $\mathbb{F}_2^{2}$ via the \emph{phase space representation},
\begin{equation}
X \mapsto (1,0),\ Z \mapsto (0,1),\ Y \mapsto (1,1), 
\end{equation}
which extends coordinate-wise to define a group isomorphism $\mathsf{P}^n \to \mathbb{F}_2^{2n}$. Thus, an element of $\mathsf{P}^n$ is represented by $n$ ``$X$-bits'' and $n$ ``$Z$-bits''. Explicitly, $X^{x_1} Z^{z_1} \otimes \dots \otimes X^{x_n}Z^{z_n}$ is mapped to $(x_1,\dots,x_n,\, z_1,\dots,z_n)^T$.
For example
\begin{equation}
      X \otimes I \otimes Z \otimes Y \mapsto (1,0,0,1,\, 0,0,1,1)^{T} \, . 
\end{equation}
This identification will allow for the application of Boolean Fourier analysis to the Pauli group.
We denote the operation of swapping the $X$-bits and $Z$-bits by a bar, $\overline{(x_1,\dots,x_n,\, z_1,\dots,z_n)} = (z_1,\dots,z_n,\, x_1,\dots,x_n)$.
The symplectic inner product \eqref{eq:def:SympProduct-Paulis} then corresponds to
\begin{equation}
\inprodp{ e}{ e\primed} = \overline{ e} \cdot  e\primed \, , 
\end{equation}
where the dot product is evaluated in $\mathbb F_2^{2n}$, i.e.\ modulo $2$.
We define the \emph{Pauli weight} $\wtp(e)$ of $e \in \mathsf P^n \cong \mathbb{F}_2^{2n}$ as the weight of the corresponding Pauli operator, i.e.\
$\wtp(e) = |\{i \in [n] \, : \, e_i \neq 0 \, \vee \, e_{i + n} \neq 0 \}|$.

We describe errors using Pauli channels, i.e.\ quantum channels of the form
\begin{equation}
\rho \mapsto \sum_{e \in \mathsf{P}^n} P(e)\,  e \rho e^\dagger \, ,
\end{equation}
where $P: \mathsf{P}^n \rightarrow [0,1]$ is a normalized probability distribution, the \emph{error distribution} of the channel; note again that this expression is independent of the choices of representatives of $e$ in $\mathcal{P}^n$.

A stabilizer code is defined by a set of commuting Pauli operators $g\su 1 \dots , g\su l \in \mathcal{P}^n$. They generate an abelian subgroup $\stabs \subseteq \mathcal{P}^n$, called stabilizer group, which must fulfill $-1 \not \in \mathscr{S}$. This is the analogue of the classical dual code $\stabsclass$. The codespace is defined as the simultaneous $+1$ eigenspace of the operators in $\stabs$. Standard error correction with a stabilizer code proceeds as follows. In each round all generators are measured. 
If an error $e \in \mathsf{P}^n$ occurred then the vector of all measurement outcomes can be represented by the syndrome $\syn(e)$, defined as
\begin{equation}
\syn(e)_i = \inprodp{g\su i}{e} \quad \forall i = 1,\dots,l \, .
\end{equation}
Using the measured syndrome, and based on information about the error rates, one approximates the most likely logical error and applies it as a correction. 
The outcomes of all stabilizers are determined only by the measurements of the generators (in the case of perfect measurements) via linearity of $\,\inprodp{\cdot}{e}$. 
As in the classical case, we can define the distance of a stabilizer code. However, in contrast to the classical case, there are many errors that act trivially on the code space and do not affect the logical information.
We define the \emph{distance} as the smallest weight of an undetectable error that affects the state of the logical qubit, i.e.\
\begin{equation}
d \coloneqq \min \{\wtp(e) \; : \; e \in \mathsf{P}^n \setminus \stabs, \, \syn(e) = 0 \  \} \, .
\end{equation}
Moreover, we define the \emph{pure distance} of a code as the smallest weight of any undetectable error, 
\begin{equation}
d_p \coloneqq \min \{\wtp(e) \; : \; e \in \mathsf{P}^n \setminus \{I\}, \, \syn(e) = 0 \  \} \, .
\end{equation}
The pure distance and the distance can differ significantly. For example, the distance of the toric code is equal to the lattice size. The pure distance on the other hand is $4$ independent of the lattice size, since the weight of any star or plaquette stabilizer is $4$.

In a practice, the stabilizer measurements themselves could also be faulty.
In this case, one should measure additional elements of $\stabs$ to mitigate the effect of measurement errors. This is captured by the framework of \emph{quantum data-syndrome codes}, which allow for a unified treatment of data and measurement errors \cite{ashikhmin_quantumdatasyndromecodes,fujiwara_datasyndromecodes,delfosse_beyondsingleshotfaulttolerance}.
We now give the basic definitions, following \cite{ashikhmin_quantumdatasyndromecodes}.
In the context of data-syndrome codes, errors are described by a data and a measurement part, as $e = (e_d,e_m)  \in \mathsf{P}^n \times \mathbb{F}_2^{m} \cong \mathbb{F}_2^{2n + m}$.
The swapping of $X$- and $Z$-bits now only applies to the data bits, i.e.\ $\overline{  e} = (\overline{  e_d}, e_m)$ for $e \in \mathbb{F}_2^{2n+m}$.
 We extend the symplectic product to data-syndrome codes by
\begin{equation}
\inprodds{(s_d,s_m)}{(e_d,e_m)} = \inprodp{s_d}{e_d} + s_m \cdot e_m  s_d\cdot\overline e_d + s_m \cdot e_m \, . 
\end{equation} 

A quantum data-syndrome code is defined by a stabilizer code on $n$ physical qubits with generators $g\su 1,\dots,g\su{l} \in \mathsf P^n \cong \mathbb{F}_2^{2n}$ and a classical code that encodes $l$ bits into $m$ bits.
We can always write the generator matrix of the classical code in the systematic form $G_C = \begin{bmatrix}I_l & A \end{bmatrix}$ with $A \in \mathbb{F}_2^{l \times (m-l)}$.
Instead of just the generators $g\su 1,\dots,g\su l$, we measure the stabilizers $f\su 1,\dots,f\su m$ defined by 
\begin{equation}
f\su i = \sum_{j=1}^l G_C[j,i] g\su j \, .
\end{equation}
Note that $f\su i = g\su i$ for $i \leq l$. 
We collect the stabilizers $f\su i$ into a matrix $F = \begin{bsmallmatrix}
f^{(1) \, T} \\ \vdots \\ f^{(m) \, T}
\end{bsmallmatrix}$. 
The measurements of the stabilizers can then be described by the parity check matrix 
\begin{equation}
H_{\mathrm {DS}} = \begin{bmatrix}
F & I_m 
\end{bmatrix} \, ,
\end{equation}
where the identity part describes the effect of measurement errors, as seen by the following discussion.
If an error $e = (e_d,e_m) \in \mathsf{P}^n \times \mathbb{F}_2^m$ occurred, the syndrome $\syn(e)$ is then described by
\begin{equation}
\syn(e)_i = \inprodds{h\su i}{e} \, ,
\label{eq:def:syn-quantum}
\end{equation}
where $h\su i$ is the $i$-th row of $H_{\mathrm {DS}}$. In phase space representation this simply means $\syn(e) = H \overline e = F \overline e_d + e_m$,  i.e.\ the syndrome is the sum of the ideal syndrome and the measurement errors.
We denote the row span of $H_{\mathrm {DS}}$ as $\stabs_{\mathrm {DS}}$, since it is an analogue of the stabilizer group.  

We define the Pauli weight of $e = (e_d,e_m) \in \mathsf P^n \times \mathbb{F}_2^m \cong \mathbb{F}_2^{2n + m}$ as $\wtp(e) = \wtp(e_d) + \wt(e_m)$.
Analogous to stabilizer codes, we can define the \emph{distance} of the data-syndrome code as
\begin{equation*}
d \coloneqq \min \{\wtp(e) : e \in \mathsf P^n \times \mathbb{F}_2^m \setminus \{(h,0) : h \in \stabs \}, \, \syn(e) = 0 \} \,,
\end{equation*}
and the \emph{pure distance} as 
\begin{equation*}
d_p \coloneqq \min \{\wtp(e) : e \in \mathsf P^n \times \mathbb{F}_2^m \setminus \{0\}, \, \syn(e) = 0 \} \, .
\end{equation*}
Naturally, the distance of a data-syndrome code cannot be larger than that of the underlying quantum code. 
Furthermore, it is not hard to see from the definitions that the pure distance of a data-syndrome code is the minimum of its distance
and the pure distance of the underlying stabilizer code.
Thus, the pure distance is limited primarily by the underlying quantum code.

All in all, measurement errors and data errors can be treated in a unified way, analogous to a standard stabilizer code.

\subsubsection{Identifiability results for quantum codes} 
\label{sec:Quantum-Identifiability}
Now we extend the identifiability results from classical to quantum data-syndrome codes, by applying analogous arguments to the phase space representation. 
The explicit example of the toric code is discussed in \Cref{sec:ToricCodeExample}. 

We consider a quantum data-syndrome code on $n$ qubits and $m$ measurement bits, and set $N \coloneqq 2n+m$. Similar to the classical case, we consider an error model described by independent channels on some selections of qubits and measurement bits.
We also allow that some of the channels contain only $X$- or $Z$-errors on some qubits.
Thus, in the phase space representation, we consider a set of \emph{channel supports} $\Chanset \subseteq 2^{[N]}$, and for each $\chanset \in \Chanset$ an error distribution $P_\chanset: 2^{[N]} \mapsto [0,1]$ that is only supported on $2^\chanset$, where we again identify binary vectors in $\mathbb{F}_2^N$ with subsets of $[N]$.
Since the total error is a sum of independently occurring errors, the total error distribution is again given by a Boolean convolution, 
\begin{equation}
P = \underset{\chanset \in \Chanset}{\Conv} P_\chanset \, .
\label{eq:P-convolution}
\end{equation}
We denote the sets of bits in phase space representation that only support \emph{detectable} errors as
\begin{equation}
\detectableset = \{\chanset \subseteq [N] \; : \; \syn(e) \neq 0 \; \forall e \subseteq \chanset \}
\label{eq:def:detectableset-quantum}
\end{equation}
and denote $\overline \detectableset = \{ \overline e \; : \; e \in \detectableset\}$.
Using these definitions, we can state our main result.

\begin{theorem} [General identifiability condition]
Consider a quantum data-syndrome code with $n$ qubits and $m$ measurement bits subject to noise described by the channel supports $\Chanset \subseteq 2^{[N]}$, where $N \coloneqq 2n + m$.
Assume that any union of two channel supports only supports detectable errors, i.e.\ for all $\chanset_1,\chanset_2 \in \Chanset$, $\chanset_1 \cup \chanset_2 \in \detectableset$. 
Furthermore assume $P_\chanset(0) > 0.5$ for all $\chanset \in \Chanset$.
Then the total error distribution $P$ is identifiable from the syndrome statistics.
\label{thm:main-full}
\end{theorem}

As in the classical case, we can also consider the set of ``independently occurring errors'',
\begin{equation}
\hat \Chanset = \{\emptyset \neq  e \subseteq [N] \; : \; \exists \chanset \in \Chanset \text{ such that }  e \subseteq \chanset\ \} \, ,
\end{equation}
instead of the set of channel supports $\Chanset$.
As explained \Cref{sec:Appendix-ChannelSupports},
 an equivalent condition to $\chanset_1 \cup \chanset_2 \in \detectableset$ for all $\chanset_1, \chanset_2 \in \detectableset$ in terms of these independent errors is that $\syn(e_1) \neq 0$ and $\syn(e_1 + e_2) \neq 0$ for all $e_1,e_2 \in \hat \Chanset$. Thus, we require that independently occurring errors have different syndromes. As also discussed in the classical case, this does not preclude that different errors with the same syndrome occur frequently, but they must arise as a combination of independent errors.

The most important implication of \Cref{thm:main-full} is the following.
\begin{corollary}[Pure distance and identifiability]
Consider a quantum data-syndrome code of pure distance $d_p$ on $n$ qubits and $m$ measurement bits, subject to noise described by the channel supports $\Chanset = \{ \chanset \subseteq [2n + m] \, : \, \wtp(\chanset) \leq t \}$. Furthermore assume $P_\chanset(0) > 0.5$ for all $\chanset \in \Chanset$. Then $P$ is identifiable from the syndrome statistics if $t \leq \lfloor \frac{d_p -1}{2} \rfloor$.
\begin{proof}
If the quantum data-syndrome code has pure distance $d_p$ then $\detectableset$ contains any $\chanset \subseteq [N]$ with Pauli weight at most $t$ (when viewed as a binary vector).
\end{proof}

\end{corollary}

In other words, with a code of pure distance $d_p$ we can estimate Pauli noise that is correlated across $\lfloor \frac{d_p -1}{2} \rfloor$ combined qubits and measurement bits. This proves the informal \Cref{thm:main-informal}. One could also make stronger independence assumptions, for example that data and measurement errors occur independently of each other. In this case one can consider the pure distance $d_Q$ of the underlying quantum code and the distance $d_C$ of the measurement code independently and can estimate correlations across at least $d_Q$ data qubits and $d_C$ measurement bits. Similarly, for CSS-codes, one can consider a separate $X$- and $Z$-distance if one assumes that $X$- and $Z$-errors occur independently.

Let us now discuss the proof of \Cref{thm:main-full}, which consists of carefully applying the arguments from the classical case to the phase space representation. In this framework, the main difference to the classical case is that the moments must be defined using the symplectic product instead of the normal dot product if we want them to match the measured expectation values. Thus, for any subset $a \subseteq [N]$, we define
\begin{align}
  E( a) &\coloneqq \sum_{ e \in \mathbb{F}_2^N} (-1)^{\inprodds{ a}{ e}} P( e), 
  \\
  E_\chanset ( a) &\coloneqq \sum_{ e \in \mathbb{F}_2^N} (-1)^{\inprodds{ a}{ e}} P_\chanset( e)   \, .
\end{align}
For $ s \in \stabssymp_{\mathrm {DS}}$, $E(s)$ is again exactly the expectation value of the measurement of the stabilizer in repeated rounds of error correction. However, the relation between moments and Fourier coefficients now contains a ``twist'', i.e.\
\begin{equation*}
E( a) = \sum_{ e \in \mathbb{F}_2^N} (-1)^{\inprodds{ a}{ e}} P( e) = \sum_{ e \in \mathbb{F}_2^N} (-1)^{\overline{ a}\cdot  e} P( e) = \fourier P(\overline{ a}) \, .
\end{equation*}
A similar connection between measurements and Fourier coefficients has recently been pointed out by \citet{flammia2021_paulilearningpopulationrecovery}.

Because $P_\chanset$ is only supported on $2^\chanset$, $\fourier P_{\chanset}( a) = \fourier P_{\chanset}( a \cap \chanset)$ and thus $E_{\chanset}( a) = E_{\chanset}( a \cap \overline \chanset)$. It follows that
\begin{equation}
E( a) = \prod_{\chanset \in \Chanset} E_\chanset( a \cap \overline \chanset) \, .
\label{eq:E-gamma-quantum}
\end{equation}
We can define the transformed moments $F( b)$ via the inclusion-exclusion transform \eqref{eq:F-E} as in the classical case and obtain
\begin{equation}
E( a) = \prod_{ b \subseteq a} F( b) \quad \forall a \subseteq [N] \, .
\label{eq:EfromF-quantum}
\end{equation}
By replacing $\chanset$ with $\overline \chanset$ in \Cref{lemma:RescaledMoments-1}, we obtain $F( a) = 1$ if there is no $\chanset \in \Chanset$ such that $ a \subseteq \overline{\chanset}$. 
It thus suffices to know the transformed moments $F( b)$ for $ b \in \overline{\hat \Chanset}$, where
\begin{equation*}
\overline{\hat \Chanset} = \{ \emptyset \neq  A \subseteq [N] \; : \; \exists \chanset \in \Chanset \text{ such that }  A \subseteq \overline \chanset \} \, .
\end{equation*}
The problem of learning the error rates has thus been reduced to the problem of learning the transformed moments from the measured expectation values. Explicitly, we need to solve the following equation system, which is analogous to the classical case.
\begin{equation}\label{eq:BinomialSystemE-F-Quantum}
\begin{aligned}
&\text{Given } 
(E( s))_{ s \in \stabs_{\mathrm {DS}}} 
\text{ find } 
(F( b))_{ b\in \overline{\hat \Chanset}} 
\text{ satisfying }
\\[.2em]
&E( s) = \prod_{ b \in \overline{\hat \Chanset} :  b \subseteq  s} F( b) 
\qquad 
\forall  s \in \stabs_{\mathrm {DS}} \setminus \{0\} \, .
\end{aligned}
\end{equation} 
This system can be described by the \emph{coefficient matrix} $D$, whose rows are labeled by elements of $\stabs_{\mathrm {DS}}$ and whose columns are labeled by elements of $\overline{\hat \Chanset}$, with entries
\begin{equation}
D[ s, a] = \begin{cases}
		 1 & \text{ if }  a \subseteq  s \\
		 0 & \text{ otherwise }
		 \end{cases} \, .
\label{eq:def:D-quantum}
\end{equation} 

We are now in a position to finish the proof of \Cref{thm:main-full}.
\begin{proof}[Proof of \Cref{thm:main-full}]
First we show that the coefficient matrix $D$ defined in \eqref{eq:def:D-quantum} has full rank.
Note that by the definition of $\detectableset$, if $\chanset \in \detectableset$, then also $a \in \detectableset$ for any $a \subseteq \chanset$.
The assumption that $\chanset_1 \cup \chanset_2 \in \detectableset$ for all $\chanset_1, \chanset_2 \in \Chanset$ thus implies that $ a \cup  b \in \overline \detectableset$ for all $ a,  b \in \overline{\hat \Chanset}$. It is proven in
\Cref{sec:Appendix-OrthogonalArrays} (\Cref{lemma:RandomStab-DataSyndrome})
that the elements of $\stabssymp_{\mathrm {DS}}$ look locally uniformly random on any element of $\overline \detectableset$. 
The same arguments as in the proof of \Cref{thm:main} thus imply that $D^TD$ can be re-scaled to be a principal sub-matrix of the matrix $M_t$ from \eqref{eq:def:M} (for $t = \max \{|a| \, : \, a \in \overline{\hat \Chanset}\}$). Since $M_t$ is positive-definite by \Cref{lemma:MPosDef}, we conclude that
$D^TD$ and thus $D$ is of full rank.
By \cite{chen2014_solutionsofbinomials}[Proposition 2], the re-scaled moments $F(a)$, and thus also the standard moments $E(a)$, can be estimated up to their sign.
If we assume all moments $E(a)$ to be positive, then the estimate is unique. 
A sufficient condition for all moments $E(a)$ to be positive is that $P_\chanset(0) > 0.5$ for all $\chanset \in \Chanset$.
\end{proof}

Similar to the classical case, the appearance of multiple solutions, differing by the signs of some moments, reflects the symmetries of the problem.
This becomes especially apparent when considering the simple case of a stabilizer code with single qubit noise and no measurement errors.
In this case, errors on each qubit are independent.
Thus, for a stabilizer $S = S_1 \otimes \dots \otimes S_n \in \mathcal{P}^n$, 
we have  $E(S) = \prod_{i \in [n]} E(S_i)$, which is a simpler form of the equation system \eqref{eq:BinomialSystemE-F-Quantum}.
Consider a representative $L \in \mathcal P^n$ of a logical operator (including stabilizers).
Denote as
\begin{equation}
\mathcal{A}(L) \coloneqq \{A \in \mathcal{P}^n \, : \, \wtp(A) = 1,\, \inprodp{A}{L} = 1 \}
\end{equation}
the set of single qubit Pauli operators that anti-commute with $L$.
Since $L$ commutes with all stabilizers $S$, each equation $E(S) = \prod_{i \in [n]} E(S_i)$ contains, on the right-hand side, an even number of terms $E(S_i)$ such that $S_i \in \mathcal{A}(L)$.
Thus, flipping the sign of $E(S_i)$ for all $S_i \in \mathcal{A}(L)$ does not change the measured syndrome statistics.
We see that there is one symmetry for each representative of a logical operator.

In summary, the estimation problem is best phrased in terms of moments instead of error rates. From this perspective, estimating the error distribution boils down to solving a polynomial system \eqref{eq:BinomialSystemE-F-Quantum}. If the correlations in the error distribution are small enough, this system has a finite number of discrete solutions. These solutions are described by symmetries related to the logical operators of the code. Under the additional mild assumption that the error rates are smaller than $\frac 12$ the solution is unique.

\subsection{Practical estimation}
\label{sec:PracticalEstimation}
While the main result of this paper lies in establishing identifiability conditions in principle, our proofs also suggest a practical method to actually perform the estimation. Here, we briefly sketch this method and heuristically comment on its sample complexity. A detailed analysis is ongoing work.

We suggest a method of moments estimator, i.e.\ to 
use empirical expectation values $\hat E(s)$ for $E(s)$ in ~\eqref{eq:BinomialSystemE-F-Quantum}. 
The basic steps (for the quantum case, but the classical case is exactly analogous) are the following:
\begin{enumerate}
\item Perform $m$ syndrome measurements, and use them to compute the empirical expectation value $\hat E(s)$ for each stabilizer $s$.
\item Insert these empirical expectation values into \eqref{eq:BinomialSystemE-F-Quantum} and solve the resulting binomial system to obtain an estimate $(\hat F(b))_{b\in \overline{\hat \Chanset}}$ of the transformed moments.
\item Insert this estimate into \eqref{eq:EfromF-quantum} to obtain an estimate of the moments.
\item Perform the inverse Fourier transform \eqref{eq:FTinv} to obtain an estimate $\hat P$ of the Pauli error rates.
\end{enumerate}

The relevant binomial system can be solved e.g. using the methods described by \citet{chen2014_solutionsofbinomials}.
It will generally be over-determined.
In principle, one can select a subset of equations such that the system is exactly determined, and then solve it analytically, resulting in a closed form expression for the estimate of the error rates in terms of the empirically measured expectation values.
This is illustrated by the example of the toric code (\Cref{sec:ToricCodeExample}), for which our method reproduces the estimator suggested by \citet{obrien_adaptiveweightestimator}. 
This estimator has also been applied by\citet{varbanov_leakagedetectionhmm}. 
However, since the empirical expectations contain a certain error, this solution might not always yield a proper probability distribution, if the number of samples used is low. 
Furthermore, for an over-determined system, selecting a subset of equations removes some of the measured information, which can reduce the accuracy of the estimate. In such cases it might be preferable to instead use a least-squares solver for the over-determined system.

Since our algorithm is designed for an on-line setting where only syndrome information is available, the sample complexity must necessarily be worse than that of algorithms using arbitrary measurements, such as \cite{flammia2020_efficientestimationofpaulichannels} and \cite{flammia2021_paulilearningpopulationrecovery}.
Each syndrome measurement contains a measurement of each stabilizer generator, therefore the expectation values of all stabilizers can be estimated simultaneously and no large measurement overhead is needed. 
However, the estimation error has to be propagated through the binomial system, the inclusion-exclusion transform and the inverse Fourier transform.
Heuristically, we expect a sample complexity similar to that of factor graph learning \cite{abbeel_learningfactorgraphs}, which was applied to Pauli channel estimation by \citet[Result 3]{flammia2020_efficientestimationofpaulichannels}. However, there will be a complicating factor accounting for the conditioning of the binomial system described by the coefficient matrix $D$.
Compared to other algorithms designed for the on-line setting \cite{fujiwara_instantaneouschannelestimationcss,
fowler_scalableextractionoferrormodelsfromqec,
huo_2017,
wootton_qiskitbenchmarking,
florjanczykbrun_insituadaptiveencoding,
obrien_adaptiveweightestimator,
combes_insitucharofdevice, wagner2021_optimalnoiseestimationfromsyndromes}, 
we expect that our method compares favorably, since these algorithms are either designed for a limited class of codes and Pauli noise models, deal with likelihood functions based on syndrome probabilities, not moments, which quickly become intractable, or only work in the limit of vanishing error rates.

In practice, it will not be necessary to estimate the expectation values of all stabilizers, but only a limited subset will be sufficient, depending on how many equation of the binomial system are retained. We expect that it is sufficient to only consider a selection of neighboring stabilizers for each qubit. For example, in the toric code example \Cref{sec:ToricCodeExample}, only three expectation values are needed per qubit, independent of the system size.

Finally, in case that the Pauli noise model is not strictly identifiable, we expect that the estimate will combine the rates of indistinguishable errors, but still give a good estimate of these total error rates.

\section{Discussion}
In this work, we considered the estimation of Pauli channels just from the measurements done anyway during the decoding of a quantum code.
We established a general condition for the feasibility of this estimation and explained the relation to the pure distance of the code.
Essentially, the estimation is possible as long as the noise has no  correlations which exceed the detection capabilities of the code.
This result does not rely on the limit of vanishing error rates, and applies even if high weight errors occur frequently, as long as these high weight errors arise as a combination of independent lower weight errors.
Our results cover general stabilizer codes and quantum data-syndrome codes \cite{ashikhmin_quantumdatasyndromecodes,fujiwara_datasyndromecodes, delfosse_beyondsingleshotfaulttolerance} and also take into account measurement errors.
The previously proposed \emph{adaptive weight estimator} \cite{obrien_adaptiveweightestimator} can be seen as a special case of our general results, since it solves a specific instance of the system \eqref{eq:BinomialSystemE-F} for single qubit noise and a limited class of codes, those that admit a minimum weight matching decoder.

An interesting new problem is an analogue of our results on the logical level.
Since we cannot identify the rates of undetectable errors, there is no full identifiability for correlations larger than the pure distance.
However, we expect that the logical channel can still be identified, as long as there are no correlations larger than the actual distance of the code.
This would practically allow for an analogue of \Cref{thm:main-informal}, using the distance instead of the pure distance. 
Furthermore, our proof naturally suggests a method of moments estimator for the Pauli error rates by inserting the empirical moments in \eqref{eq:BinomialSystemE-F}. A detailed analysis of this method, including rigorous performance guarantees as well as a better assessment of its sample complexity, is ongoing work.

\section*{Acknowledgments}
We thank Daniel Miller for discussions on algebraic geometry.
This work was funded by the Deutsche Forschungsgemeinschaft (DFG, German Research Foundation) under Germany's Excellence Strategy – Cluster of Excellence Matter and Light for Quantum Computing (ML4Q) EXC 2004/1 – 390534769.
The work of M.K.\ is also supported by the DFG via the Emmy Noether program (grant number 441423094) and by the German Federal Ministry of Education and Research (BMBF) within the funding program ``quantum technologies -- from basic research to market'' in the joint project MIQRO (grant number 13N15522).

%

\section*{Competing interests}

The authors declare no competing interest.

\appendix
\section*{Appendix}
\addcontentsline{toc}{section}{Appendix}
\renewcommand{\thesubsection}{\Alph{section}.\arabic{subsection}}
In this appendix, we provide some auxiliary statements in order to keep this work largely self-contained. 
We also provide the proof of \Cref{lemma:MPosDef} which was omitted in the main text.
Finally, we explain how the toric code example in \Cref{sec:ToricCodeExample} relates to the results derived by \citet{obrien_adaptiveweightestimator}.

\section{Proof of \texorpdfstring{\Cref{lemma:RescaledMoments-1}}{Lemma 2}}
\label{sec:Appendix-ProofOfLemma:RescaledMoments-1}
Inserting equation \eqref{eq:E-gamma} into equation \eqref{eq:F-E} yields
\begin{align}
F( a) &= \prod_{ c \subseteq  a} (\prod_{\chanset \in \Chanset} E_\chanset( c \cap \chanset))^{\mu( a, c)} = \prod_{\chanset \in \Chanset}\prod_{ c \subseteq  a} E_\chanset( c \cap \chanset)^{\mu( a, c)} \nonumber
\\
&=    \prod_{\chanset \in \Chanset}
      \prod_{ b \subseteq \chanset} 
      \prod_{ c \subseteq  a :  c \cap \chanset =  b} 
      E_\chanset( b)^{\mu( a, c)} = 
\prod_{\chanset \in \Chanset}\prod_{ b \subseteq \chanset} E_\chanset( b)^{ \sum_{ c \subseteq  a:  c \cap \chanset =  b} \mu( a, c)} \, .
\label{eq:F(A)}
\end{align}
Writing $ c=  b \cup  r$ with $ r \cap \chanset  = \emptyset$, and using the fact that the sum is empty if $ b \not \subseteq  a$, we obtain
\begin{align*}
\sum_{ c \subseteq  a :  c \cap \chanset =  b} \mu( a, c) &= (-1)^{| a|}\sum_{ c \subseteq  a :  c \cap \chanset =  b} (-1)^{| c|} \\
&=(-1)^{| a|}(-1)^{| b|} \,\iverson{ b \subseteq  a} \sum_{ c \subseteq  a \setminus \chanset} (-1)^{| c|} = \mu( a, b)\, \iverson{ b \subseteq  a} \sum_{0 \leq a \leq | a\setminus \chanset|} \binom{| a\setminus \chanset|}{a}(-1)^a \\
&= \mu( a, b) \,\iverson{ b \subseteq  a} \,\iverson{ a \setminus \chanset = \emptyset} = \mu( a, b) \,\iverson{ b\subseteq  a} \,\iverson{ a \subseteq  \chanset} \, .
\end{align*}
Substituting this identity into \eqref{eq:F(A)} yields the desired expression. 

\section{Equivalent condition on channel supports}
\label{sec:Appendix-ChannelSupports}
Below Theorems~\ref{thm:main} and \ref{thm:main-full}, we stated two equivalent characterizations of detectable errors in terms of channel supports. 
The following lemma formalizes this statement. 
The proof is the same in both the classical and the quantum case.

\begin{lemma}
Let $\detectableset$ be as defined in equation \eqref{eq:def:Delta}, $\Chanset \subseteq 2^{[n]}$ and let $\hat \Chanset = \{ e \subseteq \chanset \, : \, \chanset \in \Chanset \}$. Then $\chanset_1 \cup \chanset_2 \in \detectableset$ for all $\chanset_1, \chanset_2 \in \Chanset$ if and only if $\syn(e_1) \neq 0$ and $\syn(e_1 + e_2) \neq 0$ for all $e_1,e_2 \in \hat \Chanset$.
\end{lemma}
\begin{proof}
Note that by the definition of $\detectableset$, 
if $\chanset \in \detectableset$ then $e \in \detectableset$ for all $e \subseteq \chanset$. 
Thus, $\gamma_1 \cup \gamma_2 \in \detectableset$ implies that $e_1,e_2 \in \detectableset$ for all $e_1 \subseteq \chanset_1$ and $e_2 \subseteq \chanset_2$ and, in particular, $\syn(e_1) \neq 0$. 
Furthermore, for any such $e_1,e_2$, we have that $e_1 + e_2 \subseteq \chanset_1 \cup \chanset_2$, since the addition as binary vectors corresponds to the symmetric difference of the indicator sets. 
This implies, in particular, $\syn(e_1 + e_2) \neq 0$. 

The other direction of the equivalence is proven similarly by noting that any subset of $e \subseteq \chanset_1 \cup \chanset_2$ can be written as $e_1 + e_2$ for appropriate choices of $e_1 \subseteq \chanset_1$ and $e_2 \subseteq \chanset_2$, where it could be that either $e_1$ or $e_2$ is the empty set.
\end{proof}

\section{Orthogonal array properties}
\label{sec:Appendix-OrthogonalArrays}
In this section, 
we prove local randomness properties of stabilizer codes. 
Consider a quantum data-syndrome code with stabilizers $\stabs_\mathrm{DS}$ on $n$ qubits and $m$ measurement bits. 
Set $N \coloneqq 2n + m$. 
By $ T$ we denote the $|\stabs_\mathrm{DS}|\times N$-matrix formed by all elements in the span of the rows of the parity check matrix $ H_\mathrm{DS}$, i.e.\ the rows of $ T$ are exactly the elements of $\stabs_\mathrm{DS}$. 
Then, the following local randomness property holds.
\begin{lemma}
Let $\chanset \in \detectableset$. In the restriction $ T^{|\overline \chanset}$ of $ T$ to columns in $\chanset$, every bit-string appears equally often as a row.
\label{lemma:RandomStab-DataSyndrome}
\end{lemma}
We adapt the proof of \cite[Theorem 3.29]{hedayat2012_orthogonalarrays} to our situation. 
\begin{proof}
We denote the restriction of a matrix $ M$ to the columns indexed by $\chanset \subseteq [n]$ by $ M^{|\chanset}$.
By the definition \eqref{eq:def:syn-quantum} of a syndrome and the definition ( equation \eqref{eq:def:detectableset-quantum}) of $\detectableset$ it follows that for any row $h$ of $H_\mathrm{DS}$ and any $e \subseteq \chanset$, $\inprodds{h_i}{e} = 0$ if and only if $e = 0$ (as a binary vector). 
Since $\inprodds{h_i}{e} = \overline h_i \cdot e$, we conclude that the columns of $ H_\mathrm{DS}^{|\overline \chanset}$ are linearly independent (over $\mathbb{F}_2$), i.e.\ $ H_\mathrm{DS}^{|\overline \chanset}$ is of rank $|\overline \chanset|$.
The rows of $ T^{|\overline \chanset}$ are, by definition of $ T$, the vectors $ \zeta^T  H_\mathrm{DS}^{|\overline \chanset}$ for $ \zeta \in \mathbb{F}_2^l$, where $l$ is the number of rows of $ H$. 
The number of times a bit-string $ z$ appears as a row in $ T$ is thus equal to the number of vectors $ \zeta \in \mathbb{F}_2^l$ with $ \zeta^T  H^{|\overline \chanset} =  z$. Since $ H^{|\overline \chanset}$ has rank $|\overline \chanset|$, this number is $2^{l - |\overline \chanset|}$ independent of~$z$.
\end{proof}

\Cref{lemma:RandomStab-Classical} for classical codes is a direct corollary, which follows by only considering the classical (measurement) part of the data-syndrome code.

\section{Analysis of the intersection matrix}
In this appendix, 
we first summarize the definition and the most important properties of the Schur complement (see e.g.\ the book by Horn and Johnson \cite{horn2012_matrixanalysis}).
On this basis we then prove \Cref{lemma:MPosDef} on the intersection matrix.

\subsection{Iterated Schur complements}
\label{sec:SchurComplements}
For a block matrix 
\begin{equation}
M_2 = \begin{pmatrix}
M_{11} & M_{12} \\
M_{21} & M_{22}
\end{pmatrix}, 
\end{equation}
with $M_{11}$ being invertible, the \emph{Schur complement} of $M_{11}$ in $M_2$ is defined as
\begin{equation}\label{eq:SchurComplement}
M_2/M_{11} = M_{22} - M_{21}M_{11}^{-1}M_{12} \, .
\end{equation}
If $M_{11}$ is invertible, then $M_2$ has full rank if and only if $M_2/M_{11}$ has full rank. 
Furthermore, then $M_2$ is positive definite if and only if $M_{11}$ and $M_2/M_{11}$ are positive definite.
For the extended block matrix
\begin{equation}
M_3 = \begin{pmatrix}
M_{11} & M_{12} & M_{13} \\
M_{21} & M_{22} & M_{23} \\
M_{31} & M_{32} & M_{33}
\end{pmatrix}  
\end{equation}
the upper left block of $M_3/M_{11}$ is given by $M_2/M_{11}$ and the following \emph{quotient property of Schur complements} holds:
\begin{equation}
M_3 / M_2 = (M_3/ M_{11}) / (M_2 / M_{11})\, .
\label{eq:SchurQuotientProperty}
\end{equation}

\subsection{The intersection matrix and proof of \texorpdfstring{\Cref{lemma:MPosDef}}{the lemma}}
\label{sec:RankM}
In the  \Cref{sec:OrthogonalArry-maintext}, we have shown that the coefficient matrix $D$ has full rank, which implies identifiability of the error distribution.
The only thing missing is the proof of \Cref{lemma:MPosDef}, which was omitted.
We will now supply this proof, i.e.\ we show that the intersection matrix $M_t$ defined in \eqref{eq:def:M} is positive-definite for any $t \in \mathbb{N}$.
 
Ordering the indexing sets $ a, b \in \Chanset_{\leq t}$ by their size, for $r \leq t$, we can view $M_r$ as a sub-matrix in the upper left corner of $M_t$.
$M_t$ is an $(\alpha_t -1)\times (\alpha_t -1)$ matrix, where we define $\alpha_t = \sum_{k = 0}^t \binom{n}{k}$. 
Note that in the case $t = 1$, $M$ takes a very simple form. Since two single-element sets are either equal or disjoint, $M_1$ has entries $2$ on the diagonal and $1$ elsewhere, i.e.\
\begin{equation}
M_1 = I + \allones{1} \, ,
\end{equation}
where $\mathbf{1}_t$ denotes the all ones vector of dimension $\binom{n}{t}$. 
This matrix is clearly positive-definite. 
In general, $M_t$ does not admit such a simple expression. 
However, we will show that certain \emph{Schur complements} of $M_t$ can always be expressed in this simple form (see \Cref{sec:SchurComplements}
 for the required background on Schur complements).

\begin{lemma}
\label{lemma:Schur-Complement-M-Simple}
Let $M_t$ be as in \eqref{eq:def:M}. Then the Schur complement $M_t/M_{t-1}$ of $M_{t-1}$ in $M_t$ is well-defined and given by
\begin{equation}
M_t/M_{t-1} = I + \frac{\allones{t}}{\alpha_{t-1}} \, ,
\end{equation}
where $\mathbf{1}_t$ is the all ones vector of dimension $\binom{n}{t}$ and $\alpha_t = \sum_{k = 0}^t \binom{n}{k}$.
For $t = 1$ we set $M_1 / M_0 \coloneqq M_1$. 
\end{lemma}

\Cref{lemma:MPosDef} is a simple corollary of this statement:
\begin{proof}[Proof of \Cref{lemma:MPosDef}]
Clearly, $M_t / M_{t-1}$ is positive-definite. Since a matrix is positive-definite if the upper left block and its Schur complement are both positive-definite, it follows
that $M_t$ is positive-definite.
\end{proof}

For the proof of \Cref{lemma:Schur-Complement-M-Simple}, some further preliminaries and notation are needed.
We can view $M_t$ as a block matrix consisting of blocks $(M_{w,w\primed}\ : w,w\primed \in \{1,\ldots,t\})$, where $M_{w,w\primed}$ labels the block whose rows are indexed by sets of size $w$ and whose columns are indexed by sets of size $w\primed$.
By repeatedly applying the quotient property of Schur complements  (\eqref{eq:SchurQuotientProperty} in the Appendix),
we can express $M_t / M_{t-1}$ as a chain of Schur complements:
\begin{align}
M_t / M_{t-1} & = (M_t / M_{t-2}) / (M_{t-1} / M_{t-2}) \nonumber \\
& = [(M_t / M_{t-3}) / (M_{t-2} / M_{t-3})]/(M_{t-1}/M_{t-2}) \nonumber \\ 
& = [[[(M_t/M_1)/(M_2/M_1)] / (M_3/M_1)] \ldots] /(M_{t-1} / M_{t-2}) \nonumber \, .
\end{align}
This chain is described by the recursive definitions
\begin{align}
&M_t\su 0 = M_t \, , 
\\
&M_t\su 1 = (M_t / M_1) \, ,
\\
&M_t\su i = M_t\su{i-1} / (M_i / M_{i-1}) \, .
\label{eq:def:M_tsusi}
\end{align}
We use the same labeling of blocks for $M_t\su i$ as for $M_t$, i.e.\ $M\su i_{w,w\primed}$ is the block of $M\su i_t$ whose rows are labeled by sets of weight $w$ and whose columns are labeled by sets of weight $w\primed$, i.e.\
\begin{equation}
M\su i_{w,w\primed} = [M_t\su{i-1}]_{w,w\primed} \, .
\end{equation}

Of course, this is only well-defined if such a block actually exists in $M\su i_t$, i.e.\ if $w,w\primed > i$. In this case neither the entries nor the dimensions of $M\su i_{w,w\primed}$ depend on $t$, so the index t is indeed not needed.

The following useful observations follow directly from the properties of the Schur complement:
\begin{lemma}
The matrix $M_t$ defined in \eqref{eq:def:M} and the matrices $M_t\su i$ defined in \eqref{eq:def:M_tsusi} have the following properties:
\
\begin{enumerate}
\item $M_t\su i = M_t / {M_i}$ \label{prop:Msusi-Schur }
\item $M_{i}/M_{i-1} = M\su{i-1}_{i,i}$  

\item $M_{w,i}\su{i-1} = (M_{i,w}\su{i-1})^T$
\item For any $w,w\primed > i$:

 $M_{w,w\primed}\su i = M_{w,w\primed}\su{i-1} - M\su{i-1}_{w,i}(M_i/ M_{i-1})^{-1}M\su{i-1}_{i,w\primed}$ \label{prop:Msusi-block}
\end{enumerate}
\label{lemma:Schur-properties-list}
\end{lemma}
\begin{proof}[Proof of \Cref{lemma:Schur-properties-list}]
\
\begin{enumerate}
\item
For $i = 0,1$ the statement holds by definition (interpreting $M_t\su i / M_0 = M_t \su i$). 
The statement then follows by induction, using the quotient property \eqref{eq:SchurQuotientProperty}.

\item 
We denote with $M_{>i-1}$ the sub-matrix of $M_t$ whose rows and columns are labeled by sets of size $>i-1$. Furthermore, $M_{>i-1,i-1}$ denotes the sub-matrix whose rows are labeled by sets of size $>i-1$ and whose columns are labeled with sets of size $i-1$. 
The matrix $M_{i-1,>i-1}$ is defined analogously 
and for any matrix we use a subscript $_{i,i}$ to denote the block whose rows and columns are labeled by sets of weight $i$. 
Remember the notation $M\su{i-1}_{i,i} \coloneqq [M\su{i-1}_t]_{i,i}$.
Then,
\begin{align*}
M\su{i-1}_{i,i} 
&= [M_t / M_{i-1}]_{i,i} 
\\
&= 
[M_{>i-1} -  M_{>i-1,i-1}M_{i-1}^{-1} 
M_{i-1,>i-1}]_{i,i} \\
& =  M_{i,i} - [M_{>i-1,i-1}M_{i-1}^{-1}M_{i-1,>i-1}]_{i,i} \\
&=M_{i,i} - M_{i,i-1} M_{i-1}^{-1} M_{i-1,i} 
\\
&= M_i / M_{i-1} \, .
\end{align*}

\item This statement holds because $M_t$ is symmetric.

\item The calculation is analogous to the one above. Using property \ref{prop:Msusi-Schur } and the definition of Schur complements \eqref{eq:SchurComplement}, we obtain
\begin{align*}
M\su{i}_{w,w\primed} &= [M_{>i}\su{i-1} - M_{>i,i}\su{i-1}(M_{i,i}\su{i-1})^{-1}M_{i,>i}\su{i-1}]_{w,w\primed} \\
&= M\su i_{w,w\primed} - M_{w,i}\su{i-1}(M_{i,i}\su{i-1})^{-1}M_{i,w\primed}\su{i-1}\, .
\end{align*}
We substitute $M_{i,i}\su{i-1}$ with $M_i / M_{i-1}$, since we use this form later.
\end{enumerate}
\end{proof} 

\Cref{lemma:Schur-Complement-M-Simple} is now a direct consequence of the following more explicit statement. 

\begin{lemma}
\label{lemma:Schur-Complement-M-Full}
$M_t\su{i}$ is given entry-wise by
\begin{equation}
M_t\su{i}[ a, b] =  f\su{i+1}(| a\cap  b|) + u\su{i}(| a|,| b|) \, ,
\label{eq:Msusi-final}
\end{equation}
where
\begin{align}
&f\su{i}(x) = \sum_{k = i}^x \binom{x}{k} = 2^x - \sum_{k=0}^{i-1}\binom{x}{k} \\
&u\su{i}(w,w\primed) = \frac{\binom{w-1}{i}\binom{w\primed-1}{i}}{\alpha_i} \label{eq:Ususi-Definition} \\
&\alpha_i = \sum_{k = 0}^{i} \binom{n}{k}
\end{align}
\end{lemma}
To recover \Cref{lemma:Schur-Complement-M-Simple}, notice that $f\su{i}(x)= \iverson{x = i}$ if $x \leq i$ and $u\su{i-1}(i,i) = \frac{1}{\alpha_{i-1}}$. Thus we obtain
\begin{equation}
\begin{split}
(M_i / M_{i-1})[ a, b] &= M\su{i-1}_{i,i}[ a, b] \\ &= \iverson{| a \cap  b| =  i} + u\su{i-1}(i,i) \\ &= \iverson{ a =  b}+ \frac{1}{\alpha_{i-1}} \, ,
\end{split}
\label{eq:Mi/i-1}
\end{equation}
where we used that $| a| = | b|= i$ because we are considering the block indexed by $i, i$. Choosing $i = t$ yields \Cref{lemma:Schur-Complement-M-Simple}.

\begin{proof}[Proof of \Cref{lemma:Schur-Complement-M-Full}]
We compute $M_t\su{i}$ block-wise. 
Remember that the rows of $M_{w,w\primed}\su{i-1}$ are indexed by subsets of $[n]$ of weight $w$ and the columns by subsets of $[n]$ of size $w\primed$, and that $M_t\su{i-1} = M_t / M_{i-1}$ only contains blocks $M_{w,w\primed}\su{i-1}$ with $w,w\primed \geq i$. 

We prove the lemma by induction. 
The base case $i = 0$ is follows directly as
\begin{equation}
M_{w,w\primed}\su{0}[ a, b] = M_{w,w\primed}[ a, b] 
= 2^{| a\cap  b|} = 2^{| a\cap  b|} -1 +1 
=f\su{1}(| a\cap  b|) + u\su 0(w,w\primed) \, .
\end{equation}
For the induction step we assume that \eqref{eq:Msusi-final} holds for $i-1$ and we will prove it for $i$. 
Using the induction hypothesis, we apply \eqref{eq:Mi/i-1} to obtain
\begin{equation}
M_{i}/M_{i-1} 
= 
M\su{i-1}_{i,i} 
=
I + \frac{\allones{i}}{\alpha_{i-1}} \, ,
\end{equation}
where $\mathbf{1}_i$ is again the all ones vector of length $\binom{n}{i}$. The inverse of this matrix can be computed via the Sherman-Morrison formula to be 
\begin{equation}
(M_{i}/M_{i-1})^{-1} = I - \frac{\allones{i}}{\alpha_{i-1} + \mathbf{1}_i^T\mathbf{1}^{}_i} = I - \frac{\allones{i}}{\alpha_{i-1} + \binom{n}{i}} = I - \frac{\allones{i}}{\alpha_{i}} \, .
\end{equation}
Combining this expression with property \ref{prop:Msusi-block} in \Cref{lemma:Schur-properties-list}, we obtain
\begin{equation}
M_{w,w\primed}\su i = M_{w,w\primed}\su{i-1} - M\su{i-1}_{w,i}M\su{i-1}_{i,w\primed} 
+ \frac{M\su{i-1}_{w,i}\allones{i}M\su{i-1}_{i,w\primed}}{\alpha_{i}} \, .
\label{eq:MwwExpansion}
\end{equation}
Calculating this term by term, starting with the last term, we obtain
\begin{align*}
(\mathbf{1}_i^T M\su{i-1}_{i,w\primed})[ a] &= \sum_{| c| = i} M\su{i-1}_{i,w\primed}[ c, a]  
= \sum_{| c| = i} (f\su{i}(| a\cap  c|) + u\su{i-1}(i,w\primed)) \\
&= \sum_{| c| = i} (\iverson{ c \subseteq  a} + u\su{i-1}(i,w\primed)) 
= \binom{w\primed}{i} +\binom{n}{i}u\su{i-1}(i,w\primed) \, ,
\end{align*}
where $\sum_{| c| = i}$ indicates a sum over all subsets of $[n]$ of size $i$. The third equality used that
\begin{equation}
f\su{i}(| a\cap  c|) = \iverson{| a\cap  c| = i} =\iverson{ c\subseteq  a}
\end{equation}
 since $| c| = i$ and the last equality used $| a|= w\primed$. Defining
\begin{equation}
S\su{i}(w\primed) = \binom{w\primed}{i} +\binom{n}{i}u\su{i-1}(i,w\primed)
\label{eq_Ssusi-Definition}
\end{equation} 
and using the symmetry of $M_{i,w}\su {i-1}$ we obtain
\begin{equation}
M\su{i-1}_{w,i}\allones{i}M\su{i-1}_{i,w\primed} = S\su{i}(w)S\su{i}(w\primed)\allones{i} \, .
\label{eq:M11M}
\end{equation}
Next, we consider the second term of \eqref{eq:MwwExpansion}, 
\begin{align}
(M\su{i-1}_{w,i}M\su{i-1}_{i,w\primed})[ a, b] \nonumber 
&= \sum_{| c| = i} M_{w,i}\su{i-1}[ a, c]M_{i,w\primed}\su{i-1}[ c, b] \nonumber \\
&= \sum_{| c| = i}[f\su i (| a\cap  c|) +u\su{i-1}(i,w)] 
[f\su i (|  c\cap  b|) +u\su{i-1}(i,w\primed)] \nonumber \\
&= \sum_{| c| = i}\iverson{ c\subseteq  a}\iverson{ c \subseteq  b}
 + u\su{i-1}(i,w)\sum_{| c| = i}\iverson{ c \subseteq  b}
\nonumber \\ &\phantom{=} 
+u\su{i-1}(i,w\primed)\sum_{| c| = i}\iverson{ c\subseteq  a}
 + \sum_{| c| = i}u\su{i-1}(i,w)u\su{i-1}(i,w\primed) \nonumber \\
&=\binom{| a\cap  b|}{i} + u\su{i-1}(i,w)\binom{w\primed}{i} \nonumber
+u\su{i-1}(i,w\primed)\binom{w}{i} \nonumber 
 +\binom{n}{i}u\su{i-1}(i,w)u\su{i-1}(i,w\primed) \nonumber \\
&=\binom{| a\cap  b|}{i}+T\su{i}(w,w\primed) \, ,
\label{eq:MM}
\end{align}
where we defined 
\begin{equation}
T\su{i}(w,w\primed) 
= u\su{i-1}(i,w)\binom{w\primed}{i} 
+u\su{i-1}(i,w\primed)\binom{w}{i}  +\binom{n}{i}u\su{i-1}(i,w)u\su{i-1}(i,w\primed) \, .
\label{eq:Tsusi-Definition}
\end{equation}
Thus, combining \eqref{eq:M11M}, \eqref{eq:MM}, and the induction assumption \eqref{eq:Msusi-final} with \eqref{eq:MwwExpansion} we obtain
\begin{align}
&M_{w,w\primed}\su i[ a, b] = f\su{i}(| a\cap  b|) + u\su{i-1}(w,w\primed) - \binom{| a \cap  b|}{i} - T\su i (w,w\primed) +\frac{S\su i (w)S\su i (w\primed)}{\alpha_i} \nonumber  \\
&= f\su{i+1}(| a\cap  b|) +\frac{S\su i (w)S\su i (w\primed) + \alpha_i(u\su{i-1}(w,w\primed) - T\su i (w,w\primed))}{\alpha_i} \label{eq:Mww-ST} \, ,
\end{align}
where we used $f\su{i+1}(| a\cap  b|)=f\su i (| a\cap  b|)- \binom{| a \cap  b|}{i}$.
Finally, we calculate the last term. By direct calculation from the definitions \eqref{eq:Tsusi-Definition}, \eqref{eq_Ssusi-Definition} we obtain
\begin{equation}
S\su i (w)S\su i (w\primed) - \binom{n}{i}T\su i (w,w\primed) = \binom{w}{i}\binom{w\primed}{i}\, ,
\label{eq:Ssusi-Tsusi}
\end{equation}
and thus,
\begin{align}
&S\su i (w)S\su i (w\primed) + \alpha_i(u\su{i-1}(w,w\primed) - T\su i (w,w\primed)) \nonumber \\
&= S\su i (w)S\su i (w\primed) - \binom{n}{i}T\su i (w,w\primed) + \alpha_{i-1}u\su{i-1}(w,w\primed) + \binom{n}{i}u\su{i-1}(w,w\primed) -\alpha_{i-1}T\su i (w,w\primed)\, ,
\label{eq:M_w-FinalEnumerator}
\end{align}
where we have used that $\alpha_i = \binom{n}{i} + \alpha_{i-1}$. 
Inserting the definitions of $u\su i$ and $T\su i$ from \eqref{eq:Ususi-Definition} and \eqref{eq:Tsusi-Definition}, we obtain $\alpha_{i-1}u\su{i}(w,w\primed) = \binom{w-1}{i-1}\binom{w\primed-1}{i-1}$ and $\binom{n}{i}u\su{i-1}(w,w\primed) -\alpha_{i-1}T\su i (w,w\primed) = - \binom{w-1}{i-1}\binom{w\primed}{i} -\binom{w\primed-1}{i-1}\binom{w}{i}$. Together with \eqref{eq:Ssusi-Tsusi} this yields
\begin{align*}
\eqref{eq:M_w-FinalEnumerator} &= \binom{w}{i}\binom{w\primed}{i} + \binom{w-1}{i-1}\binom{w\primed-1}{i-1}  - \binom{w-1}{i-1}\binom{w\primed}{i} -\binom{w\primed-1}{i-1}\binom{w}{i} \\
&= \left[\binom{w}{i}- \binom{w-1}{i-1}\right]\left[\binom{w\primed}{i}- \binom{w\primed-1}{i-1}\right] \\
&=\binom{w-1}{i}\binom{w\primed-1}{i} \, ,
\end{align*}
where the last step is based on Pascals identity.
Substituting this result back into \eqref{eq:Mww-ST} we obtain
\begin{equation}
M_{w,w\primed}\su i[ a, b] = f\su{i+1}(| a \cap  b|)+\frac{\binom{w-1}{i}\binom{w\primed-1}{i}}{\alpha_i} \, ,
\end{equation}
which finishes the proof.
\end{proof}

\section{Connection to Adaptive Weight Estimator}
\label{sec:AppendixSO-Solution}
Here, we show that the solution for the toric code derived in \Cref{sec:ToricCodeExample} coincides with the solution given by \citet{obrien_adaptiveweightestimator}.

The solution derived in \Cref{sec:ToricCodeExample} is 
\begin{equation}
E(Z_4) = \pm \sqrt{\frac{E(S_1)E(S_2)}{E(S_1S_2)}} \, .
\label{eq:SolutionToric-2}
\end{equation}
On the other hand, the solution given by \citet[eq. (14)]{obrien_adaptiveweightestimator} is
\begin{equation}
p_4 
= \frac 12 \mp \sqrt{\frac 14 - \frac{P(S_1 = S_2=1) - P(S_1 = 1)P(S_2 = 1)}{1-2P(S_1 \neq S_2)}} \, .
\label{eq:SO-Solution}
\end{equation}
Here, $P$ is used to denote the probabilities of events under random sampling of the errors. 
Note that in \cite{obrien_adaptiveweightestimator} the result is phrased in terms of expectation values of the stabilizer outcomes and errors viewed as taking values $0$ or $1$, which directly translates into the probabilities above.

To see that these two solutions coincide, first notice that $E(Z_4) = 1-2p_4$. Thus, \eqref{eq:SolutionToric-2} can be rewritten as
\begin{equation}
p_4 = \frac{1}{2} \mp \sqrt{\frac{E(S_1)E(S_2)}{4E(S_1S_2)}} \, .
\end{equation}
Similarly, we have $E(S_i) = 1-2P(S_i = 1)$ and 
$E(S_1S_2) = 1-2P(S_1 \neq S_2)$. 
Using these equations, equality of the solutions can be shown as follows:
\begin{widetext}
\begin{equation*}
\begin{split}
\frac{E(S_1)E(S_2)}{4E(S_1S_2)} &= \frac{(1-2P(S_1 = 1))(1-2P(S_2 = 1))}{4(1-2P(S_1 \neq S_2))} \\
&= \frac{1-2P(S_1 = 1)-2P(S_2 = 1)+4P(S_1 = 1)P(S_2 = 1) + 1 - 2P(S_1 \neq S_2) - 1 + 2P(S_1 \neq S_2)}{4(1-2P(S_1 \neq S_2))} \\
&= \frac 14 - \frac{2P(S_1 = 1)+2P(S_2 = 1)-4P(S_1 = 1)P(S_2 = 1) - 2P(S_1 \neq S_2)}{4(1-2P(S_1 \neq S_2))} \\
&= \frac 14 - \frac{P(S_1 = S_2 = 1) - P(S_1=1)P(S_2 = 1)}{1-2P(S_1 \neq S_2)} \, , \\
\end{split}
\end{equation*} 
\end{widetext}
where in the last equality we used that
\begin{align*}
&P(S_1 = 1) + P(S_2 = 1) - P(S_1 \neq S_2) \\
&= P(S_1 = 1,S_2 = 1) + P(S_1 = 1,S_2 = 0) \\
& \phantom{=} + P(S_1 = 1, S_2 = 1) + P(S_1 = 0,S_2 = 1) \\
&\phantom{=}- P(S_1 =1, S_2 = 0) - P(S_1 = 0, S_2 = 1) \\
&= 2P(S_1 = S_2 = 1) \, .
\end{align*}
This shows that the two solutions \eqref{eq:SolutionToric-2}  and \eqref{eq:SO-Solution} are indeed equivalent.

\bibliographystyle{./myapsrev4-2}
\bibliography{bibliography}

\begin{thebibliography}{40}%
\makeatletter
\providecommand \@ifxundefined [1]{%
 \@ifx{#1\undefined}
}%
\providecommand \@ifnum [1]{%
 \ifnum #1\expandafter \@firstoftwo
 \else \expandafter \@secondoftwo
 \fi
}%
\providecommand \@ifx [1]{%
 \ifx #1\expandafter \@firstoftwo
 \else \expandafter \@secondoftwo
 \fi
}%
\providecommand \natexlab [1]{#1}%
\providecommand \enquote  [1]{``#1''}%
\providecommand \bibnamefont  [1]{#1}%
\providecommand \bibfnamefont [1]{#1}%
\providecommand \citenamefont [1]{#1}%
\providecommand \href@noop [0]{\@secondoftwo}%
\providecommand \href [0]{\begingroup \@sanitize@url \@href}%
\providecommand \@href[1]{\@@startlink{#1}\@@href}%
\providecommand \@@href[1]{\endgroup#1\@@endlink}%
\providecommand \@sanitize@url [0]{\catcode `\\12\catcode `\$12\catcode
  `\&12\catcode `\#12\catcode `\^12\catcode `\_12\catcode `\%12\relax}%
\providecommand \@@startlink[1]{}%
\providecommand \@@endlink[0]{}%
\providecommand \url  [0]{\begingroup\@sanitize@url \@url }%
\providecommand \@url [1]{\endgroup\@href {#1}{\urlprefix }}%
\providecommand \urlprefix  [0]{URL }%
\providecommand \Eprint[0]{\href }%
\providecommand \doibase [0]{https://doi.org/}%
\providecommand \selectlanguage [0]{\@gobble}%
\providecommand \bibinfo  [0]{\@secondoftwo}%
\providecommand \bibfield  [0]{\@secondoftwo}%
\providecommand \translation [1]{[#1]}%
\providecommand \BibitemOpen [0]{}%
\providecommand \bibitemStop [0]{}%
\providecommand \bibitemNoStop [0]{.\EOS\space}%
\providecommand \EOS [0]{\spacefactor3000\relax}%
\providecommand \BibitemShut  [1]{\csname bibitem#1\endcsname}%
\let\auto@bib@innerbib\@empty
\bibitem [{\citenamefont {Robertson}\ \emph {et~al.}(2017)\citenamefont
  {Robertson}, \citenamefont {Granade}, \citenamefont {Bartlett},\ and\
  \citenamefont {Flammia}}]{robertson2017_tailoredcodessmallmemories}%
  \BibitemOpen
  \bibfield  {author} {\bibinfo {author} {\bibfnamefont {A.}~\bibnamefont
  {Robertson}}, \bibinfo {author} {\bibfnamefont {C.}~\bibnamefont {Granade}},
  \bibinfo {author} {\bibfnamefont {S.~D.}\ \bibnamefont {Bartlett}},\ and\
  \bibinfo {author} {\bibfnamefont {S.~T.}\ \bibnamefont {Flammia}},\ }\bibinfo
  {title} {\emph {Tailored codes for small quantum memories}},\ \href
  {https://doi.org/10.1103/PhysRevApplied.8.064004} {\bibfield  {journal}
  {\bibinfo  {journal} {Phys. Rev. Applied}\ }\textbf {\bibinfo {volume} {8}},\
  \bibinfo {pages} {064004} (\bibinfo {year} {2017})}\BibitemShut {NoStop}%
\bibitem [{\citenamefont {Florjanczyk}\ and\ \citenamefont
  {Brun}(2016)}]{florjanczykbrun_insituadaptiveencoding}%
  \BibitemOpen
  \bibfield  {author} {\bibinfo {author} {\bibfnamefont {J.}~\bibnamefont
  {Florjanczyk}}\ and\ \bibinfo {author} {\bibfnamefont {T.~A.}\ \bibnamefont
  {Brun}},\ }\href {https://doi.org/10.48550/ARXIV.1612.05823} {\bibinfo
  {title} {\emph {In-situ adaptive encoding for asymmetric quantum error
  correcting codes}}} (\bibinfo {year} {2016})\BibitemShut {NoStop}%
\bibitem [{\citenamefont {Bonilla~Ataides}\ \emph {et~al.}(2021)\citenamefont
  {Bonilla~Ataides}, \citenamefont {Tuckett}, \citenamefont {Bartlett},
  \citenamefont {Flammia},\ and\ \citenamefont
  {Brown}}]{ataides2021_XZZXSurfaceCode}%
  \BibitemOpen
  \bibfield  {author} {\bibinfo {author} {\bibfnamefont {J.~P.}\ \bibnamefont
  {Bonilla~Ataides}}, \bibinfo {author} {\bibfnamefont {D.~K.}\ \bibnamefont
  {Tuckett}}, \bibinfo {author} {\bibfnamefont {S.~D.}\ \bibnamefont
  {Bartlett}}, \bibinfo {author} {\bibfnamefont {S.~T.}\ \bibnamefont
  {Flammia}},\ and\ \bibinfo {author} {\bibfnamefont {B.~J.}\ \bibnamefont
  {Brown}},\ }\bibinfo {title} {\emph {The {XZZX} surface code}},\ \href
  {https://doi.org/10.1038/s41467-021-22274-1} {\bibfield  {journal} {\bibinfo
  {journal} {Nat. Commun.}\ }\textbf {\bibinfo {volume} {12}},\ \bibinfo
  {pages} {2172} (\bibinfo {year} {2021})}\BibitemShut {NoStop}%
\bibitem [{\citenamefont {Higgott}(2021)}]{higgott2021_pymatching}%
  \BibitemOpen
  \bibfield  {author} {\bibinfo {author} {\bibfnamefont {O.}~\bibnamefont
  {Higgott}},\ }\href {https://doi.org/10.48550/ARXIV.2105.13082} {\bibinfo
  {title} {\emph {Pymatching: A python package for decoding quantum codes with
  minimum-weight perfect matching}}} (\bibinfo {year} {2021})\BibitemShut
  {NoStop}%
\bibitem [{\citenamefont {{Dennis}}\ \emph {et~al.}(2002)\citenamefont
  {{Dennis}}, \citenamefont {{Kitaev}}, \citenamefont {{Landahl}},\ and\
  \citenamefont {{Preskill}}}]{dennis_topologicalquantummemory}%
  \BibitemOpen
  \bibfield  {author} {\bibinfo {author} {\bibfnamefont {E.}~\bibnamefont
  {{Dennis}}}, \bibinfo {author} {\bibfnamefont {A.}~\bibnamefont {{Kitaev}}},
  \bibinfo {author} {\bibfnamefont {A.}~\bibnamefont {{Landahl}}},\ and\
  \bibinfo {author} {\bibfnamefont {J.}~\bibnamefont {{Preskill}}},\ }\bibinfo
  {title} {\emph {Topological quantum memory}},\ \href
  {https://doi.org/10.1063/1.1499754} {\bibfield  {journal} {\bibinfo
  {journal} {J. Math. Phys.}\ }\textbf {\bibinfo {volume} {43}},\ \bibinfo
  {pages} {4452} (\bibinfo {year} {2002})},\
  \Eprint{https://arxiv.org/abs/quant-ph/0110143} {arXiv:quant-ph/0110143
  [quant-ph]}\BibitemShut {NoStop}%
\bibitem [{\citenamefont {Nickerson}\ and\ \citenamefont
  {Brown}(2019)}]{nickersonbrown_analyisngcorrelatednoisesrfacecodeadaptvedecoding}%
  \BibitemOpen
  \bibfield  {author} {\bibinfo {author} {\bibfnamefont {N.~H.}\ \bibnamefont
  {Nickerson}}\ and\ \bibinfo {author} {\bibfnamefont {B.~J.}\ \bibnamefont
  {Brown}},\ }\bibinfo {title} {\emph {Analysing correlated noise on the
  surface code using adaptive decoding algorithms}},\ \href
  {https://doi.org/10.22331/q-2019-04-08-131} {\bibfield  {journal} {\bibinfo
  {journal} {{Quantum}}\ }\textbf {\bibinfo {volume} {3}},\ \bibinfo {pages}
  {131} (\bibinfo {year} {2019})}\BibitemShut {NoStop}%
\bibitem [{\citenamefont {Spitz}\ \emph {et~al.}(2018)\citenamefont {Spitz},
  \citenamefont {Tarasinski}, \citenamefont {Beenakker},\ and\ \citenamefont
  {O'Brien}}]{obrien_adaptiveweightestimator}%
  \BibitemOpen
  \bibfield  {author} {\bibinfo {author} {\bibfnamefont {S.~T.}\ \bibnamefont
  {Spitz}}, \bibinfo {author} {\bibfnamefont {B.}~\bibnamefont {Tarasinski}},
  \bibinfo {author} {\bibfnamefont {C.~W.~J.}\ \bibnamefont {Beenakker}},\ and\
  \bibinfo {author} {\bibfnamefont {T.~E.}\ \bibnamefont {O'Brien}},\ }\bibinfo
  {title} {\emph {Adaptive weight estimator for quantum error correction in a
  time-dependent environment}},\ \href
  {https://doi.org/http://dx.doi.org/10.1002/qute.201870015} {\bibfield
  {journal} {\bibinfo  {journal} {Advanced Quantum Technologies}\ }\textbf
  {\bibinfo {volume} {1}},\ \bibinfo {pages} {1870015} (\bibinfo {year}
  {2018})}\BibitemShut {NoStop}%
\bibitem [{\citenamefont {{Babar}}\ \emph {et~al.}(2015)\citenamefont
  {{Babar}}, \citenamefont {{Botsinis}}, \citenamefont {{Alanis}},
  \citenamefont {{Ng}},\ and\ \citenamefont {{Hanzo}}}]{babar_ldpcreview}%
  \BibitemOpen
  \bibfield  {author} {\bibinfo {author} {\bibfnamefont {Z.}~\bibnamefont
  {{Babar}}}, \bibinfo {author} {\bibfnamefont {P.}~\bibnamefont {{Botsinis}}},
  \bibinfo {author} {\bibfnamefont {D.}~\bibnamefont {{Alanis}}}, \bibinfo
  {author} {\bibfnamefont {S.~X.}\ \bibnamefont {{Ng}}},\ and\ \bibinfo
  {author} {\bibfnamefont {L.}~\bibnamefont {{Hanzo}}},\ }\bibinfo {title}
  {\emph {Fifteen years of quantum {LDPC} coding and improved decoding
  strategies}},\ \href {https://doi.org/10.1109/ACCESS.2015.2503267} {\bibfield
   {journal} {\bibinfo  {journal} {IEEE Access}\ }\textbf {\bibinfo {volume}
  {3}},\ \bibinfo {pages} {2492} (\bibinfo {year} {2015})}\BibitemShut
  {NoStop}%
\bibitem [{\citenamefont {Huang}\ \emph {et~al.}(2020)\citenamefont {Huang},
  \citenamefont {Newman},\ and\ \citenamefont
  {Brown}}]{huang2020_weightedunionfind}%
  \BibitemOpen
  \bibfield  {author} {\bibinfo {author} {\bibfnamefont {S.}~\bibnamefont
  {Huang}}, \bibinfo {author} {\bibfnamefont {M.}~\bibnamefont {Newman}},\ and\
  \bibinfo {author} {\bibfnamefont {K.~R.}\ \bibnamefont {Brown}},\ }\bibinfo
  {title} {\emph {Fault-tolerant weighted union-find decoding on the toric
  code}},\ \bibfield  {journal} {\bibinfo  {journal} {Physical Review A}\
  }\textbf {\bibinfo {volume} {102}},\ \href
  {https://doi.org/10.1103/physreva.102.012419} {10.1103/physreva.102.012419}
  (\bibinfo {year} {2020})\BibitemShut {NoStop}%
\bibitem [{\citenamefont {Chubb}(2021)}]{chub2022_tensornetworkdecoder}%
  \BibitemOpen
  \bibfield  {author} {\bibinfo {author} {\bibfnamefont {C.~T.}\ \bibnamefont
  {Chubb}},\ }\href {https://doi.org/10.48550/ARXIV.2101.04125} {\bibinfo
  {title} {\emph {General tensor network decoding of 2d pauli codes}}}
  (\bibinfo {year} {2021})\BibitemShut {NoStop}%
\bibitem [{\citenamefont {Darmawan}\ and\ \citenamefont
  {Poulin}(2018)}]{darmawan2018_tensornetworkdecoder}%
  \BibitemOpen
  \bibfield  {author} {\bibinfo {author} {\bibfnamefont {A.~S.}\ \bibnamefont
  {Darmawan}}\ and\ \bibinfo {author} {\bibfnamefont {D.}~\bibnamefont
  {Poulin}},\ }\bibinfo {title} {\emph {Linear-time general decoding algorithm
  for the surface code}},\ \bibfield  {journal} {\bibinfo  {journal} {Physical
  Review E}\ }\textbf {\bibinfo {volume} {97}},\ \href
  {https://doi.org/10.1103/physreve.97.051302} {10.1103/physreve.97.051302}
  (\bibinfo {year} {2018})\BibitemShut {NoStop}%
\bibitem [{\citenamefont {Wallman}\ and\ \citenamefont
  {Emerson}(2016)}]{wallman_randomizedcompiling}%
  \BibitemOpen
  \bibfield  {author} {\bibinfo {author} {\bibfnamefont {J.~J.}\ \bibnamefont
  {Wallman}}\ and\ \bibinfo {author} {\bibfnamefont {J.}~\bibnamefont
  {Emerson}},\ }\bibinfo {title} {\emph {Noise tailoring for scalable quantum
  computation via randomized compiling}},\ \href
  {https://doi.org/10.1103/PhysRevA.94.052325} {\bibfield  {journal} {\bibinfo
  {journal} {Phys. Rev. A}\ }\textbf {\bibinfo {volume} {94}},\ \bibinfo
  {pages} {052325} (\bibinfo {year} {2016})}\BibitemShut {NoStop}%
\bibitem [{\citenamefont {Ware}\ \emph {et~al.}(2021)\citenamefont {Ware},
  \citenamefont {Ribeill}, \citenamefont {Rist\`e}, \citenamefont {Ryan},
  \citenamefont {Johnson},\ and\ \citenamefont
  {da~Silva}}]{ware2021_randomizedcompiling-experimental}%
  \BibitemOpen
  \bibfield  {author} {\bibinfo {author} {\bibfnamefont {M.}~\bibnamefont
  {Ware}}, \bibinfo {author} {\bibfnamefont {G.}~\bibnamefont {Ribeill}},
  \bibinfo {author} {\bibfnamefont {D.}~\bibnamefont {Rist\`e}}, \bibinfo
  {author} {\bibfnamefont {C.~A.}\ \bibnamefont {Ryan}}, \bibinfo {author}
  {\bibfnamefont {B.}~\bibnamefont {Johnson}},\ and\ \bibinfo {author}
  {\bibfnamefont {M.~P.}\ \bibnamefont {da~Silva}},\ }\bibinfo {title} {\emph
  {Experimental {Pauli}-frame randomization on a superconducting qubit}},\
  \href {https://doi.org/10.1103/PhysRevA.103.042604} {\bibfield  {journal}
  {\bibinfo  {journal} {Phys. Rev. A}\ }\textbf {\bibinfo {volume} {103}},\
  \bibinfo {pages} {042604} (\bibinfo {year} {2021})}\BibitemShut {NoStop}%
\bibitem [{\citenamefont {Beale}\ \emph {et~al.}(2018)\citenamefont {Beale},
  \citenamefont {Wallman}, \citenamefont {Guti\'errez}, \citenamefont {Brown},\
  and\ \citenamefont {Laflamme}}]{beale2018_qecdecoheresnoise}%
  \BibitemOpen
  \bibfield  {author} {\bibinfo {author} {\bibfnamefont {S.~J.}\ \bibnamefont
  {Beale}}, \bibinfo {author} {\bibfnamefont {J.~J.}\ \bibnamefont {Wallman}},
  \bibinfo {author} {\bibfnamefont {M.}~\bibnamefont {Guti\'errez}}, \bibinfo
  {author} {\bibfnamefont {K.~R.}\ \bibnamefont {Brown}},\ and\ \bibinfo
  {author} {\bibfnamefont {R.}~\bibnamefont {Laflamme}},\ }\bibinfo {title}
  {\emph {Quantum error correction decoheres noise}},\ \href
  {https://doi.org/10.1103/PhysRevLett.121.190501} {\bibfield  {journal}
  {\bibinfo  {journal} {Phys. Rev. Lett.}\ }\textbf {\bibinfo {volume} {121}},\
  \bibinfo {pages} {190501} (\bibinfo {year} {2018})}\BibitemShut {NoStop}%
\bibitem [{\citenamefont {Flammia}\ and\ \citenamefont
  {O'Donnell}(2021)}]{flammia2021_paulilearningpopulationrecovery}%
  \BibitemOpen
  \bibfield  {author} {\bibinfo {author} {\bibfnamefont {S.~T.}\ \bibnamefont
  {Flammia}}\ and\ \bibinfo {author} {\bibfnamefont {R.}~\bibnamefont
  {O'Donnell}},\ }\bibinfo {title} {\emph {Pauli error estimation via
  population recovery}},\ \href {https://doi.org/10.22331/q-2021-09-23-549}
  {\bibfield  {journal} {\bibinfo  {journal} {Quantum}\ }\textbf {\bibinfo
  {volume} {5}},\ \bibinfo {pages} {549} (\bibinfo {year} {2021})}\BibitemShut
  {NoStop}%
\bibitem [{\citenamefont {Harper}\ \emph {et~al.}(2021)\citenamefont {Harper},
  \citenamefont {Yu},\ and\ \citenamefont
  {Flammia}}]{harper2020_sparsepauliestimation}%
  \BibitemOpen
  \bibfield  {author} {\bibinfo {author} {\bibfnamefont {R.}~\bibnamefont
  {Harper}}, \bibinfo {author} {\bibfnamefont {W.}~\bibnamefont {Yu}},\ and\
  \bibinfo {author} {\bibfnamefont {S.~T.}\ \bibnamefont {Flammia}},\ }\bibinfo
  {title} {\emph {Fast estimation of sparse quantum noise}},\ \href
  {https://doi.org/10.1103/PRXQuantum.2.010322} {\bibfield  {journal} {\bibinfo
   {journal} {PRX Quantum}\ }\textbf {\bibinfo {volume} {2}},\ \bibinfo {pages}
  {010322} (\bibinfo {year} {2021})}\BibitemShut {NoStop}%
\bibitem [{\citenamefont {Flammia}\ and\ \citenamefont
  {Wallman}(2020)}]{flammia2020_efficientestimationofpaulichannels}%
  \BibitemOpen
  \bibfield  {author} {\bibinfo {author} {\bibfnamefont {S.~T.}\ \bibnamefont
  {Flammia}}\ and\ \bibinfo {author} {\bibfnamefont {J.~J.}\ \bibnamefont
  {Wallman}},\ }\bibinfo {title} {\emph {Efficient estimation of {Pauli}
  channels}},\ \bibfield  {journal} {\bibinfo  {journal} {ACM Transactions on
  Quantum Computing}\ }\textbf {\bibinfo {volume} {1}},\ \href
  {https://doi.org/10.1145/3408039} {10.1145/3408039} (\bibinfo {year}
  {2020})\BibitemShut {NoStop}%
\bibitem [{\citenamefont {Harper}\ \emph {et~al.}(2020)\citenamefont {Harper},
  \citenamefont {Flammia},\ and\ \citenamefont
  {Wallman}}]{harper2020_efficientlearningofquantumnoise}%
  \BibitemOpen
  \bibfield  {author} {\bibinfo {author} {\bibfnamefont {R.}~\bibnamefont
  {Harper}}, \bibinfo {author} {\bibfnamefont {S.~T.}\ \bibnamefont
  {Flammia}},\ and\ \bibinfo {author} {\bibfnamefont {J.~J.}\ \bibnamefont
  {Wallman}},\ }\bibinfo {title} {\emph {Efficient learning of quantum
  noise}},\ \href {https://doi.org/10.1038/s41567-020-0992-8} {\bibfield
  {journal} {\bibinfo  {journal} {Nat. Phys.}\ }\textbf {\bibinfo {volume}
  {16}},\ \bibinfo {pages} {1184} (\bibinfo {year} {2020})}\BibitemShut
  {NoStop}%
\bibitem [{\citenamefont
  {Fujiwara}(2014)}]{fujiwara_instantaneouschannelestimationcss}%
  \BibitemOpen
  \bibfield  {author} {\bibinfo {author} {\bibfnamefont {Y.}~\bibnamefont
  {Fujiwara}},\ }\href {https://doi.org/10.48550/ARXIV.1405.6267} {\bibinfo
  {title} {\emph {Instantaneous quantum channel estimation during quantum
  information processing}}} (\bibinfo {year} {2014})\BibitemShut {NoStop}%
\bibitem [{\citenamefont {Fowler}\ \emph {et~al.}(2014)\citenamefont {Fowler},
  \citenamefont {Sank}, \citenamefont {Kelly}, \citenamefont {Barends},\ and\
  \citenamefont {Martinis}}]{fowler_scalableextractionoferrormodelsfromqec}%
  \BibitemOpen
  \bibfield  {author} {\bibinfo {author} {\bibfnamefont {A.~G.}\ \bibnamefont
  {Fowler}}, \bibinfo {author} {\bibfnamefont {D.}~\bibnamefont {Sank}},
  \bibinfo {author} {\bibfnamefont {J.}~\bibnamefont {Kelly}}, \bibinfo
  {author} {\bibfnamefont {R.}~\bibnamefont {Barends}},\ and\ \bibinfo {author}
  {\bibfnamefont {J.~M.}\ \bibnamefont {Martinis}},\ }\href
  {https://doi.org/10.48550/ARXIV.1405.1454} {\bibinfo {title} {\emph {Scalable
  extraction of error models from the output of error detection circuits}}}
  (\bibinfo {year} {2014})\BibitemShut {NoStop}%
\bibitem [{\citenamefont {Huo}\ and\ \citenamefont {Li}(2017)}]{huo_2017}%
  \BibitemOpen
  \bibfield  {author} {\bibinfo {author} {\bibfnamefont {M.-X.}\ \bibnamefont
  {Huo}}\ and\ \bibinfo {author} {\bibfnamefont {Y.}~\bibnamefont {Li}},\
  }\bibinfo {title} {\emph {Learning time-dependent noise to reduce logical
  errors: real time error rate estimation in quantum error correction}},\ \href
  {https://doi.org/10.1088/1367-2630/aa916e} {\bibfield  {journal} {\bibinfo
  {journal} {New J. Phys.}\ }\textbf {\bibinfo {volume} {19}},\ \bibinfo
  {pages} {123032} (\bibinfo {year} {2017})}\BibitemShut {NoStop}%
\bibitem [{\citenamefont {Wootton}(2020)}]{wootton_qiskitbenchmarking}%
  \BibitemOpen
  \bibfield  {author} {\bibinfo {author} {\bibfnamefont {J.~R.}\ \bibnamefont
  {Wootton}},\ }\bibinfo {title} {\emph {Benchmarking near-term devices with
  quantum error correction}},\ \href {https://doi.org/10.1088/2058-9565/aba038}
  {\bibfield  {journal} {\bibinfo  {journal} {Quantum Science and Technology}\
  }\textbf {\bibinfo {volume} {5}},\ \bibinfo {pages} {044004} (\bibinfo {year}
  {2020})}\BibitemShut {NoStop}%
\bibitem [{\citenamefont {Combes}\ \emph {et~al.}(2014)\citenamefont {Combes},
  \citenamefont {Ferrie}, \citenamefont {Cesare}, \citenamefont {Tiersch},
  \citenamefont {Milburn}, \citenamefont {Briegel},\ and\ \citenamefont
  {Caves}}]{combes_insitucharofdevice}%
  \BibitemOpen
  \bibfield  {author} {\bibinfo {author} {\bibfnamefont {J.}~\bibnamefont
  {Combes}}, \bibinfo {author} {\bibfnamefont {C.}~\bibnamefont {Ferrie}},
  \bibinfo {author} {\bibfnamefont {C.}~\bibnamefont {Cesare}}, \bibinfo
  {author} {\bibfnamefont {M.}~\bibnamefont {Tiersch}}, \bibinfo {author}
  {\bibfnamefont {G.~J.}\ \bibnamefont {Milburn}}, \bibinfo {author}
  {\bibfnamefont {H.~J.}\ \bibnamefont {Briegel}},\ and\ \bibinfo {author}
  {\bibfnamefont {C.~M.}\ \bibnamefont {Caves}},\ }\href
  {https://doi.org/10.48550/ARXIV.1405.5656} {\bibinfo {title} {\emph {In-situ
  characterization of quantum devices with error correction}}} (\bibinfo {year}
  {2014})\BibitemShut {NoStop}%
\bibitem [{\citenamefont {Wagner}\ \emph {et~al.}(2021)\citenamefont {Wagner},
  \citenamefont {Kampermann}, \citenamefont {Bru\ss{}},\ and\ \citenamefont
  {Kliesch}}]{wagner2021_optimalnoiseestimationfromsyndromes}%
  \BibitemOpen
  \bibfield  {author} {\bibinfo {author} {\bibfnamefont {T.}~\bibnamefont
  {Wagner}}, \bibinfo {author} {\bibfnamefont {H.}~\bibnamefont {Kampermann}},
  \bibinfo {author} {\bibfnamefont {D.}~\bibnamefont {Bru\ss{}}},\ and\
  \bibinfo {author} {\bibfnamefont {M.}~\bibnamefont {Kliesch}},\ }\bibinfo
  {title} {\emph {Optimal noise estimation from syndrome statistics of quantum
  codes}},\ \href {https://doi.org/10.1103/PhysRevResearch.3.013292} {\bibfield
   {journal} {\bibinfo  {journal} {Phys. Rev. Research}\ }\textbf {\bibinfo
  {volume} {3}},\ \bibinfo {pages} {013292} (\bibinfo {year}
  {2021})}\BibitemShut {NoStop}%
\bibitem [{\citenamefont {Kelly}\ \emph {et~al.}(2016)\citenamefont {Kelly},
  \citenamefont {Barends}, \citenamefont {Fowler}, \citenamefont {Megrant},
  \citenamefont {Jeffrey}, \citenamefont {White}, \citenamefont {Sank},
  \citenamefont {Mutus}, \citenamefont {Campbell}, \citenamefont {Chen},
  \citenamefont {Chen}, \citenamefont {Chiaro}, \citenamefont {Dunsworth},
  \citenamefont {Lucero}, \citenamefont {Neeley}, \citenamefont {Neill},
  \citenamefont {O'Malley}, \citenamefont {Quintana}, \citenamefont {Roushan},
  \citenamefont {Vainsencher}, \citenamefont {Wenner},\ and\ \citenamefont
  {Martinis}}]{Kelly_ScalableInSituCalibrationDuringErrorDetection}%
  \BibitemOpen
  \bibfield  {author} {\bibinfo {author} {\bibfnamefont {J.}~\bibnamefont
  {Kelly}}, \bibinfo {author} {\bibfnamefont {R.}~\bibnamefont {Barends}},
  \bibinfo {author} {\bibfnamefont {A.~G.}\ \bibnamefont {Fowler}}, \bibinfo
  {author} {\bibfnamefont {A.}~\bibnamefont {Megrant}}, \bibinfo {author}
  {\bibfnamefont {E.}~\bibnamefont {Jeffrey}}, \bibinfo {author} {\bibfnamefont
  {T.~C.}\ \bibnamefont {White}}, \bibinfo {author} {\bibfnamefont
  {D.}~\bibnamefont {Sank}}, \bibinfo {author} {\bibfnamefont {J.~Y.}\
  \bibnamefont {Mutus}}, \bibinfo {author} {\bibfnamefont {B.}~\bibnamefont
  {Campbell}}, \bibinfo {author} {\bibfnamefont {Y.}~\bibnamefont {Chen}},
  \bibinfo {author} {\bibfnamefont {Z.}~\bibnamefont {Chen}}, \bibinfo {author}
  {\bibfnamefont {B.}~\bibnamefont {Chiaro}}, \bibinfo {author} {\bibfnamefont
  {A.}~\bibnamefont {Dunsworth}}, \bibinfo {author} {\bibfnamefont
  {E.}~\bibnamefont {Lucero}}, \bibinfo {author} {\bibfnamefont
  {M.}~\bibnamefont {Neeley}}, \bibinfo {author} {\bibfnamefont
  {C.}~\bibnamefont {Neill}}, \bibinfo {author} {\bibfnamefont {P.~J.~J.}\
  \bibnamefont {O'Malley}}, \bibinfo {author} {\bibfnamefont {C.}~\bibnamefont
  {Quintana}}, \bibinfo {author} {\bibfnamefont {P.}~\bibnamefont {Roushan}},
  \bibinfo {author} {\bibfnamefont {A.}~\bibnamefont {Vainsencher}}, \bibinfo
  {author} {\bibfnamefont {J.}~\bibnamefont {Wenner}},\ and\ \bibinfo {author}
  {\bibfnamefont {J.~M.}\ \bibnamefont {Martinis}},\ }\bibinfo {title} {\emph
  {Scalable in situ qubit calibration during repetitive error detection}},\
  \href {https://doi.org/10.1103/PhysRevA.94.032321} {\bibfield  {journal}
  {\bibinfo  {journal} {Phys. Rev. A}\ }\textbf {\bibinfo {volume} {94}},\
  \bibinfo {pages} {032321} (\bibinfo {year} {2016})}\BibitemShut {NoStop}%
\bibitem [{\citenamefont {Ashikhmin}\ \emph {et~al.}(2020)\citenamefont
  {Ashikhmin}, \citenamefont {Lai},\ and\ \citenamefont
  {Brun}}]{ashikhmin_quantumdatasyndromecodes}%
  \BibitemOpen
  \bibfield  {author} {\bibinfo {author} {\bibfnamefont {A.}~\bibnamefont
  {Ashikhmin}}, \bibinfo {author} {\bibfnamefont {C.-Y.}\ \bibnamefont {Lai}},\
  and\ \bibinfo {author} {\bibfnamefont {T.~A.}\ \bibnamefont {Brun}},\
  }\bibinfo {title} {\emph {Quantum data-syndrome codes}},\ \href
  {https://doi.org/10.1109/JSAC.2020.2968997} {\bibfield  {journal} {\bibinfo
  {journal} {IEEE Journal on Selected Areas in Communications}\ }\textbf
  {\bibinfo {volume} {38}},\ \bibinfo {pages} {449} (\bibinfo {year}
  {2020})}\BibitemShut {NoStop}%
\bibitem [{\citenamefont {{Fujiwara}}(2014)}]{fujiwara_datasyndromecodes}%
  \BibitemOpen
  \bibfield  {author} {\bibinfo {author} {\bibfnamefont {Y.}~\bibnamefont
  {{Fujiwara}}},\ }\bibinfo {title} {\emph {Ability of stabilizer quantum error
  correction to protect itself from its own imperfection}},\ \href
  {https://doi.org/10.1103/PhysRevA.90.062304} {\bibfield  {journal} {\bibinfo
  {journal} {Phys. Rev. A}\ }\textbf {\bibinfo {volume} {90}},\ \bibinfo {eid}
  {062304} (\bibinfo {year} {2014})},\ \Eprint{https://arxiv.org/abs/1409.2559}
  {arXiv:1409.2559 [quant-ph]}\BibitemShut {NoStop}%
\bibitem [{\citenamefont {Delfosse}\ \emph {et~al.}(2022)\citenamefont
  {Delfosse}, \citenamefont {Reichardt},\ and\ \citenamefont
  {Svore}}]{delfosse_beyondsingleshotfaulttolerance}%
  \BibitemOpen
  \bibfield  {author} {\bibinfo {author} {\bibfnamefont {N.}~\bibnamefont
  {Delfosse}}, \bibinfo {author} {\bibfnamefont {B.~W.}\ \bibnamefont
  {Reichardt}},\ and\ \bibinfo {author} {\bibfnamefont {K.~M.}\ \bibnamefont
  {Svore}},\ }\bibinfo {title} {\emph {Beyond single-shot fault-tolerant
  quantum error correction}},\ \href {https://doi.org/10.1109/tit.2021.3120685}
  {\bibfield  {journal} {\bibinfo  {journal} {{IEEE} Transactions on
  Information Theory}\ }\textbf {\bibinfo {volume} {68}},\ \bibinfo {pages}
  {287} (\bibinfo {year} {2022})}\BibitemShut {NoStop}%
\bibitem [{\citenamefont {{Zia}}\ \emph {et~al.}(2007)\citenamefont {{Zia}},
  \citenamefont {{Reilly}},\ and\ \citenamefont
  {{Shirani}}}]{zia_ldpcerrorrateestimation}%
  \BibitemOpen
  \bibfield  {author} {\bibinfo {author} {\bibfnamefont {A.}~\bibnamefont
  {{Zia}}}, \bibinfo {author} {\bibfnamefont {J.~P.}\ \bibnamefont
  {{Reilly}}},\ and\ \bibinfo {author} {\bibfnamefont {S.}~\bibnamefont
  {{Shirani}}},\ }\bibfield  {title} {\bibinfo {title} {\emph {Distributed
  parameter estimation with side information: A factor graph approach}},\ }in\
  \href {https://doi.org/10.1109/ISIT.2007.4557603} {\emph {\bibinfo
  {booktitle} {2007 IEEE International Symposium on Information Theory}}}\
  (\bibinfo {year} {2007})\ pp.\ \bibinfo {pages} {2556--2560}\BibitemShut
  {NoStop}%
\bibitem [{\citenamefont
  {O'Donnell}(2014)}]{odonnel2014_analysisofbooleanfunctions}%
  \BibitemOpen
  \bibfield  {author} {\bibinfo {author} {\bibfnamefont {R.}~\bibnamefont
  {O'Donnell}},\ }\href {https://doi.org/10.1017/CBO9781139814782} {\emph
  {\bibinfo {title} {Analysis of Boolean Functions}}}\ (\bibinfo  {publisher}
  {Cambridge University Press},\ \bibinfo {year} {2014})\BibitemShut {NoStop}%
\bibitem [{\citenamefont {Mao}\ and\ \citenamefont
  {Kschischang}(2005)}]{mao2005_convolutionalfactorgraphs}%
  \BibitemOpen
  \bibfield  {author} {\bibinfo {author} {\bibfnamefont {Y.}~\bibnamefont
  {Mao}}\ and\ \bibinfo {author} {\bibfnamefont {F.}~\bibnamefont
  {Kschischang}},\ }\bibinfo {title} {\emph {On factor graphs and the fourier
  transform}},\ \href {https://doi.org/10.1109/TIT.2005.846404} {\bibfield
  {journal} {\bibinfo  {journal} {IEEE Trans. Inf. Theory}\ }\textbf {\bibinfo
  {volume} {51}},\ \bibinfo {pages} {1635} (\bibinfo {year}
  {2005})}\BibitemShut {NoStop}%
\bibitem [{\citenamefont {Koller}\ and\ \citenamefont
  {Friedman}(2009)}]{koller_pgm}%
  \BibitemOpen
  \bibfield  {author} {\bibinfo {author} {\bibfnamefont {D.}~\bibnamefont
  {Koller}}\ and\ \bibinfo {author} {\bibfnamefont {N.}~\bibnamefont
  {Friedman}},\ }\href@noop {} {\emph {\bibinfo {title} {Probabilistic
  Graphical Models: Principles and Techniques - Adaptive Computation and
  Machine Learning}}}\ (\bibinfo  {publisher} {The MIT Press},\ \bibinfo {year}
  {2009})\BibitemShut {NoStop}%
\bibitem [{\citenamefont {{Aigner}}(2007)}]{Aigner2007_ACourseInEnumeration}%
  \BibitemOpen
  \bibfield  {author} {\bibinfo {author} {\bibfnamefont {M.}~\bibnamefont
  {{Aigner}}},\ }\href {https://doi.org/10.1007/978-3-540-39035-0} {\emph
  {\bibinfo {title} {A Course in Enumeration}}},\ Vol.\ \bibinfo {volume}
  {238}\ (\bibinfo  {publisher} {Springer-Verlag Berlin Heidelberg},\ \bibinfo
  {year} {2007})\BibitemShut {NoStop}%
\bibitem [{\citenamefont {{Roman}}(2006)}]{roman2006_fieldtheory}%
  \BibitemOpen
  \bibfield  {author} {\bibinfo {author} {\bibfnamefont {S.}~\bibnamefont
  {{Roman}}},\ }\href {https://doi.org/https://doi.org/10.1007/0-387-27678-5}
  {\emph {\bibinfo {title} {Field Theory}}}\ (\bibinfo  {publisher} {Springer,
  New York},\ \bibinfo {year} {2006})\BibitemShut {NoStop}%
\bibitem [{\citenamefont {Chen}\ and\ \citenamefont
  {LiTien-Yien}(2014)}]{chen2014_solutionsofbinomials}%
  \BibitemOpen
  \bibfield  {author} {\bibinfo {author} {\bibfnamefont {T.}~\bibnamefont
  {Chen}}\ and\ \bibinfo {author} {\bibnamefont {LiTien-Yien}},\ }\bibinfo
  {title} {\emph {Solutions to systems of binomial equations}},\ \href
  {https://journals.us.edu.pl/index.php/AMSIL/article/view/13987} {\bibfield
  {journal} {\bibinfo  {journal} {Annales Mathematicae Silesianae}\ }\textbf
  {\bibinfo {volume} {28}},\ \bibinfo {pages} {7} (\bibinfo {year}
  {2014})}\BibitemShut {NoStop}%
\bibitem [{\citenamefont {Hedayat}\ \emph {et~al.}(1999)\citenamefont
  {Hedayat}, \citenamefont {Sloane},\ and\ \citenamefont
  {Stufken}}]{hedayat2012_orthogonalarrays}%
  \BibitemOpen
  \bibfield  {author} {\bibinfo {author} {\bibfnamefont {A.~S.}\ \bibnamefont
  {Hedayat}}, \bibinfo {author} {\bibfnamefont {N.~J.~A.}\ \bibnamefont
  {Sloane}},\ and\ \bibinfo {author} {\bibfnamefont {J.}~\bibnamefont
  {Stufken}},\ }\href {https://doi.org/10.1007/978-1-4612-1478-6} {\emph
  {\bibinfo {title} {Orthogonal arrays: theory and applications}}}\ (\bibinfo
  {publisher} {Springer New York, NY},\ \bibinfo {year} {1999})\BibitemShut
  {NoStop}%
\bibitem [{\citenamefont
  {Delsarte}(1973)}]{delsarte1973_fundamentalcodecombinatorics}%
  \BibitemOpen
  \bibfield  {author} {\bibinfo {author} {\bibfnamefont {P.}~\bibnamefont
  {Delsarte}},\ }\bibinfo {title} {\emph {Four fundamental parameters of a code
  and their combinatorial significance}},\ \href
  {https://doi.org/https://doi.org/10.1016/S0019-9958(73)80007-5} {\bibfield
  {journal} {\bibinfo  {journal} {Information and Control}\ }\textbf {\bibinfo
  {volume} {23}},\ \bibinfo {pages} {407} (\bibinfo {year} {1973})}\BibitemShut
  {NoStop}%
\bibitem [{\citenamefont {Varbanov}\ \emph {et~al.}(2020)\citenamefont
  {Varbanov}, \citenamefont {Battistel}, \citenamefont {Tarasinski},
  \citenamefont {Ostroukh}, \citenamefont {O'Brien}, \citenamefont {DiCarlo},\
  and\ \citenamefont {Terhal}}]{varbanov_leakagedetectionhmm}%
  \BibitemOpen
  \bibfield  {author} {\bibinfo {author} {\bibfnamefont {B.~M.}\ \bibnamefont
  {Varbanov}}, \bibinfo {author} {\bibfnamefont {F.}~\bibnamefont {Battistel}},
  \bibinfo {author} {\bibfnamefont {B.~M.}\ \bibnamefont {Tarasinski}},
  \bibinfo {author} {\bibfnamefont {V.~P.}\ \bibnamefont {Ostroukh}}, \bibinfo
  {author} {\bibfnamefont {T.~E.}\ \bibnamefont {O'Brien}}, \bibinfo {author}
  {\bibfnamefont {L.}~\bibnamefont {DiCarlo}},\ and\ \bibinfo {author}
  {\bibfnamefont {B.~M.}\ \bibnamefont {Terhal}},\ }\bibinfo {title} {\emph
  {Leakage detection for a transmon-based surface code}},\ \bibfield  {journal}
  {\bibinfo  {journal} {NPJ Quantum Inf.}\ }\textbf {\bibinfo {volume} {6}},\
  \href {https://doi.org/10.1038/s41534-020-00330-w}
  {10.1038/s41534-020-00330-w} (\bibinfo {year} {2020})\BibitemShut {NoStop}%
\bibitem [{\citenamefont {Abbeel}\ \emph {et~al.}(2012)\citenamefont {Abbeel},
  \citenamefont {Koller},\ and\ \citenamefont
  {Ng}}]{abbeel_learningfactorgraphs}%
  \BibitemOpen
  \bibfield  {author} {\bibinfo {author} {\bibfnamefont {P.}~\bibnamefont
  {Abbeel}}, \bibinfo {author} {\bibfnamefont {D.}~\bibnamefont {Koller}},\
  and\ \bibinfo {author} {\bibfnamefont {A.~Y.}\ \bibnamefont {Ng}},\ }\href
  {https://doi.org/10.48550/ARXIV.1207.1366} {\bibinfo {title} {\emph {Learning
  factor graphs in polynomial time \& sample complexity}}} (\bibinfo {year}
  {2012})\BibitemShut {NoStop}%
\bibitem [{\citenamefont {Horn}\ and\ \citenamefont
  {Johnson}(2012)}]{horn2012_matrixanalysis}%
  \BibitemOpen
  \bibfield  {author} {\bibinfo {author} {\bibfnamefont {R.~A.}\ \bibnamefont
  {Horn}}\ and\ \bibinfo {author} {\bibfnamefont {C.~R.}\ \bibnamefont
  {Johnson}},\ }\href {https://doi.org/10.1017/CBO9780511810817} {\emph
  {\bibinfo {title} {Matrix Analysis}}},\ \bibinfo {edition} {2nd}\ ed.\
  (\bibinfo  {publisher} {Cambridge University Press},\ \bibinfo {year}
  {2012})\BibitemShut {NoStop}%
\end{thebibliography}%

\end{document}